\documentclass[journal]{IEEEtran}

\usepackage{lipsum} 

\usepackage{cite}
\usepackage{amsmath,amssymb,amsfonts,amsthm}
\usepackage{algorithmic}
\usepackage{graphicx}
\usepackage{textcomp}
\usepackage[table]{xcolor}
\usepackage{flushend}
\usepackage{orcidlink}
\usepackage{multirow}
\theoremstyle{plain}
\newtheorem{thm}{Theorem}
\newtheorem{lem}[thm]{Lemma}
\newtheorem{dfn}[thm]{Definition}
\newtheorem{prop}[thm]{Proposition}
\newtheorem{example}[thm]{Example}
\newtheorem{rem}[thm]{Remark}
\usepackage{booktabs}
\usepackage{makecell}
\usepackage{enumitem}
\usepackage{tabularx}
\newcolumntype{C}{>{\centering\arraybackslash}X}
\usepackage[whole]{bxcjkjatype}
\usepackage{comment}

\begin{document}

\title{
LipsAM: Lipschitz-continuous Neural Networks for Convergent Plug-and-Play Audio Signal Recovery
}

\author{Kazuki~Matsumoto,~\IEEEmembership{Student Member,~IEEE,}
        Ren~Uchida,
        Natsuki~Yoshino,
        and~Kohei~Yatabe,~\IEEEmembership{Member,~IEEE}
\thanks{This work was supported by JST FOREST Program (JPMJFR2330).}
\thanks{Kazuki Matsumoto, Ren Uchida, Natsuki Yoshino, and Kohei Yatabe are with Tokyo University of Agriculture
and Technology, Tokyo 184-8588, Japan (e-mail: \url{k_m_w_314@akane.waseda.jp}; \url{renuchida.research@gmail.com}; \url{n-yoshino@go.tuat.ac.jp}; \url{yatabe@go.tuat.ac.jp})}
}


\maketitle

\begin{abstract}

The Lipschitz continuity of deep neural networks (DNNs) is essential for establishing theoretical guarantees regarding their behavior. 
From both theoretical and practical perspectives, various methods have been proposed to construct Lipschitz-continuous architectures and control their Lipschitz constants. 
However, several DNN architectures common in audio signal processing fall outside the scope of existing theoretical frameworks, hindering the development of Lipschitz-continuous models in acoustic applications. 
In particular, despite their widespread adoption, DNNs that separately process the magnitude and phase of complex-valued signals cannot be Lipschitz continuous under existing frameworks.
In this paper, to address this limitation, we establish a theoretical foundation for constructing amplitude modifiers (AMs), a class of DNN architectures that operate solely on the magnitude of a complex-valued input, with provable Lipschitz continuity.
Specifically, we derive a necessary and sufficient condition for an AM to be Lipschitz continuous and propose LipsAMs (Lipschitz-continuous AMs) corresponding to common architectures for audio signals, including time-frequency masking.
Furthermore, we develop an efficient framework for evaluating their Lipschitz constants and analytically derive these constants for some of the proposed architectures. 
As an application, we propose CoReM-LipsAM (Controlled Residual Maps via LipsAM) for plug-and-play (PnP) audio signal recovery, integrating a DNN as a data-driven prior within a model-based signal processing algorithm. 
The convergence of the obtained PnP algorithm is structurally guaranteed by the CoReM-LipsAM architecture and empirically validated through speech dereverberation experiments.
\end{abstract}

\begin{IEEEkeywords}
Deep neural network (DNN), Lipschitz constant, plug-and-play (PnP) method, short-time Fourier transform (STFT), time-frequency masking. 
\end{IEEEkeywords}

\IEEEpeerreviewmaketitle

\section{Introduction}
\label{sec:introduction}

\IEEEPARstart{D}{espite} their widespread success across diverse signal processing applications, deep neural networks (DNNs) are often regarded as black-box systems, raising reliability concerns. 
For instance, \textit{adversarial examples}, wherein subtle input perturbations induce unexpected outputs, have been demonstrated in image recognition \cite{szegedyIntriguingPropertiesNeural2014,goodfellowExplainingHarnessingAdversarial2015} and speech recognition \cite{carliniAudioAdversarialExamples2018}, underscoring the sensitivity of DNNs. 
A primary strategy for mitigating this sensitivity is to enforce mathematically verifiable properties on DNN architectures.

Lipschitz continuity is one such property, as it directly bounds output sensitivity with respect to input perturbations.
Specifically, let $\mathbb{F}\in\{\mathbb{R},\mathbb{C}\}$, where $\mathbb{R}$ and $\mathbb{C}$ denote the sets of all real and complex numbers, respectively.
Given $L \geq 0$, a mapping $f:\mathbb{F}^N\to\mathbb{F}^M$ is \textit{$L$-Lipschitz continuous} if
\begin{align}
\forall \mathbf{x},\mathbf{y}\in\mathbb{F}^N,\quad
\|f(\mathbf{x}) - f(\mathbf{y})\|_2
\leq L \|\mathbf{x} - \mathbf{y}\|_2,
\label{eq:lipschitz_continuity}
\end{align}
where $\|\cdot\|_2$ denotes the $\ell_2$-norm.
The smallest $L$ satisfying this inequality is called the \textit{Lipschitz constant} of $f$, denoted by $\operatorname{Lip}(f)$.
When a DNN is Lipschitz continuous, one can establish certified robustness against adversarial examples \cite{cisseParsevalNetworksImproving2017,tsuzukuLipschitzMarginTrainingScalable2018}.
Furthermore, Lipschitz continuity plays a fundamental role in the theoretical analysis of neural ordinary differential equations \cite{chenNeuralOrdinaryDifferential2018} and flow-based generative models \cite{chenResidualFlowsInvertible2019}.

Lipschitz continuity is also central to the convergence analysis of Plug-and-Play (PnP) methods \cite{venkatakrishnanPlugandPlayPriorsModel2013,zhangLearningDeepCNN2017,zhangPlugandPlayImageRestoration2021,tanakaAPPLADEAdjustablePlugandPlay2022,matsumotoDeterminedBSSCombination2024,yangIntegratingPlugandPlayData2023,chanPlugandPlayADMMImage2016}.
PnP integrates a DNN as a data-driven prior into an iterative optimization algorithm for model-based signal processing.
Specifically, a proximal operator of an optimization algorithm is replaced by a denoiser, such as a pre-trained DNN.
Consequently, PnP combines the expressive power of DNNs with an interpretable mathematical model of the physical observation process.
Although originally developed for image processing tasks, including image reconstruction \cite{venkatakrishnanPlugandPlayPriorsModel2013}, deblurring \cite{zhangLearningDeepCNN2017,zhangPlugandPlayImageRestoration2021}, and super-resolution \cite{zhangLearningDeepCNN2017,zhangPlugandPlayImageRestoration2021}, PnP has also proven effective in audio signal processing tasks, such as audio declipping \cite{tanakaAPPLADEAdjustablePlugandPlay2022}, source separation \cite{matsumotoDeterminedBSSCombination2024}, and dereverberation \cite{yangIntegratingPlugandPlayData2023}.
However, integrating an arbitrary DNN into an optimization algorithm complicates theoretical analysis, rendering algorithmic convergence unverified.
To resolve this issue, various studies have imposed the Lipschitz continuity of the employed DNN, together with constraints on its Lipschitz constant, to establish theoretical convergence guarantees \cite{ryuPlugandPlayMethodsProvably2019,pesquetLearningMaximallyMonotone2021,huraultGRADIENTSTEPDENOISER2022,huraultProximalDenoiserConvergent,yukawaMonotoneLipschitzGradientDenoiser2025}.

Motivated by the importance of Lipschitz continuity, the construction of Lipschitz-continuous DNNs has been widely investigated.
These approaches fall into three primary categories: normalization-based, training-loss-based, and structure-based methods.
Normalization-based methods, such as real spectral normalization (RealSN) \cite{ryuPlugandPlayMethodsProvably2019}, constrain convolutional layers to be $1$-Lipschitz continuous by normalizing weight matrices during training.
Training-loss-based methods promote nonexpansiveness (i.e., $\operatorname{Lip}(f)\leq 1$) or related theoretical properties by penalizing the network Jacobian during training \cite{huraultProximalDenoiserConvergent,pesquetLearningMaximallyMonotone2021}.
Structure-based methods construct $1$-Lipschitz layers by employing weight matrix reparameterizations that enforce, e.g., orthogonality\cite{prach1LipschitzLayersCompared2024,meunierDynamicalSystemPerspective2022a,prachAlmostOrthogonalLayersEfficient2022,liPreventingGradientAttenuation2019}.
In parallel with network construction, the evaluation of Lipschitz constants for a given DNN remains an active area of research \cite{fazlyabEfficientAccurateEstimation2019,jordanExactlyComputingLocal2020a}.

However, existing techniques cannot directly guarantee the Lipschitz continuity of DNNs that process complex-valued \textit{spectrograms}.
These time-frequency representations, obtained via the (discrete) short-time Fourier transform (STFT), or the discrete Gabor transform (DGT), are widely used in audio signal processing.
A common strategy for processing them is to modify their magnitude and phase components separately.
This paper focuses on a class of DNN architectures, which we call \textit{amplitude modifiers (AMs)}. 
An AM $\mathcal{D}_{\!\mathcal{A}}:\mathbb{C}^N\to\mathbb{C}^N$ modifies only the magnitude of a (vectorized) complex-valued spectrogram $\mathbf{z}=[z_1,\ldots,z_N]^{\mathsf{T}}\in\mathbb{C}^N$ as follows:
\begin{equation}
\text{AM}:\mathcal{D}_{\!\mathcal{A}}(\mathbf{z})
=
\mathcal{A}(|\mathbf{z}|)\odot\operatorname{sign}(\mathbf{z}),
\label{eq:amplitude_modifier}
\end{equation}
where $\mathcal{A}:\mathbb{R}_{+}^{N}\to\mathbb{R}_{+}^{N}$ is the amplitude-modifying mapping implemented using a DNN, and $\odot$ represents element-wise multiplication.
The amplitude component $|\mathbf{z}|\in\mathbb{R}_{+}^{N}$ and the phase component $\operatorname{sign}(\mathbf{z})\in\mathbb{C}^{N}$ are given element-wise by the absolute value function and the complex-valued sign function:
\begin{align}
(|\mathbf{z}|)_{n} = |z_{n}|,\qquad  (\operatorname{sign}(\mathbf{z}))_{n} = \left\{ \begin{array}{cc}\displaystyle \frac{z_{n}}{|z_{n}|} & (z_n\neq0) \\ 0 & (z_n=0)\end{array}\right..
\end{align}
Despite their widespread use, AMs present a fundamental limitation regarding Lipschitz continuity.
Specifically, the Lipschitz continuity of the overall model $\mathcal{D}_{\!\mathcal{A}}$ is not guaranteed even when the amplitude-modifying part $\mathcal{A}$ is Lipschitz continuous, owing to the discontinuity of the sign function (see Sect.~\ref{sec:pitfall} for details).
This limitation hinders the application of Lipschitz-based theories to conventional AMs.

\begin{figure}
\includegraphics[width=\linewidth]{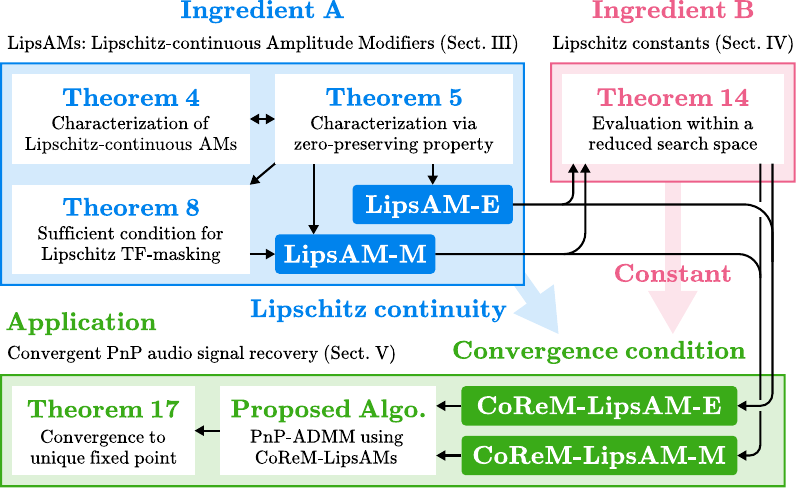}
\caption{
Overview of this paper.
White boxes denote main theoretical results or algorithms, whereas colored boxes represent the proposed DNN architectures, where ``-E'' and ``-M'' indicate direct estimation of the amplitude component and time-frequency (TF) masking, respectively.
Ingredient A establishes LipsAM (Lipschitz-continuous Amplitude Modifiers), and Ingredient B provides a framework for evaluating Lipschitz constants within a reduced search space. 
These components are integrated to construct CoReM-LipsAM (Controlled-Residual Maps via LipsAM), which satisfies the convergence conditions for a PnP (Plug-and-Play) audio signal recovery algorithm.
}
\label{fig:diagram}
\end{figure}

In this paper, to address this limitation, we establish a theoretical foundation for constructing AMs with provable Lipschitz continuity%
\footnote{
This paper extends our preliminary conference publication \cite{matsumotoLIPSAMLIPSCHITZCONTINUOUSAMPLITUDE}. 
The conference version addressed only two estimator-based LipsAM architectures, provided an incomplete proof of Theorem~\ref{thm:lipsam} by omitting the necessity condition, and lacked convergence analysis for the PnP algorithm. 
The present work resolves these theoretical limitations. 
Specifically, we provide a complete formulation of LipsAM architectures, establish both the necessary and sufficient conditions for Theorem~\ref{thm:lipsam}, and prove the convergence of the proposed PnP algorithm. 
Furthermore, this manuscript introduces ten additional formal results (including theorems, propositions, and lemmas) along with complete proofs, thereby establishing a rigorous theoretical foundation for LipsAM.
}.
As illustrated in Fig.~\ref{fig:diagram}, our proposed framework comprises two key components. 
The first component (\textbf{Ingredient~A}, shaded in blue) focuses on constructing Lipschitz-continuous AM-based DNN architectures, termed \textit{LipsAMs (Lipschitz-continuous Amplitude Modifiers)}. 
Because our approach falls into the category of structure-based methods, any DNN designed according to the proposed formulations is inherently guaranteed to be Lipschitz continuous. 
The second component (\textbf{Ingredient~B}, shaded in red) provides an efficient framework for evaluating the Lipschitz constants of LipsAMs. 
Some of these constants are derived analytically, whereas for others, they are characterized via an optimization formulation with a reduced search space. 
Together, these theoretical results lay the foundation for provably convergent PnP algorithms in audio signal processing.

As an application of the established theoretical framework and proposed LipsAM architectures, we introduce \textit{CoReM-LipsAMs (Controlled-Residual Maps via LipsAMs)}. 
In these DNN architectures, the Lipschitz constants of the residual mappings are explicitly controlled by design, thereby satisfying the convergence conditions of an alternating direction method of multipliers (ADMM)--based PnP algorithm (shaded in green in Fig.~\ref{fig:diagram}). 
We evaluated the proposed algorithm on a speech dereverberation task and empirically validated its theoretical convergence guarantees.

The primary contributions of this paper are as follows:
\begin{itemize}
\item We define LipsAMs, a class of DNNs that process the amplitude components of complex-valued inputs, and prove that satisfying the definition of LipsAM is necessary and sufficient for an amplitude-modifying DNN to be Lipschitz continuous (Definition~\ref{def:lipsam}, Theorems~\ref{thm:lipsam} and \ref{prop:zero_preserving}).

\item We derive a sufficient condition for time-frequency masking-based AMs to be Lipschitz continuous (Theorem~\ref{thm:mask_sufficient}) and propose a novel architecture satisfying this condition, extending LipsAMs to masking-based DNNs.

\item We prove that evaluation of the Lipschitz constant for any LipsAM architecture reduces to a low-dimensional optimization problem (Theorem~\ref{thm:worst_case_lipschitz_bound}) and develop an efficient framework for Lipschitz constant evaluation.

\item We propose CoReM-LipsAMs, a class of DNNs that structurally satisfy the convergence condition of the PnP algorithm (Theorem~\ref{thm:contractionTF}), thereby establishing theoretical convergence guarantees for PnP audio signal processing.

\end{itemize}

The remainder of this paper is organized as follows.
Sect.~\ref{sec:preliminary} reviews preliminaries regarding Lipschitz-continuous mappings and PnP methods.
Sects.~\ref{sec:LipsAM} and \ref{sec:bounds} detail the theoretical development of Ingredients~A and B, respectively.
Sect.~\ref{sec:algo} applies these ingredients to construct a provably convergent PnP audio signal recovery algorithm.
Sect.~\ref{sec:expt} evaluates the proposed methods through speech dereverberation experiments, and Sect.~\ref{sec:conc} concludes the paper.

\section{Lipschitz-continuity and PnP Algorithm}
\label{sec:preliminary}

This section reviews basic properties of Lipschitz-continuous mappings and PnP algorithms.
Throughout this paper, lowercase bold letters represent vectors,
e.g., $\mathbf{z}=(z_n)_{n=1}^N\in\mathbb{C}^N$, and uppercase bold letters represent matrices,
e.g., $\mathbf{A}=(a_{ij})_{i=1,\,j=1}^{I,\,J}\in\mathbb{C}^{I\times J}$, and $[N]=\{1,\ldots,N\}\subset\mathbb{N}$.

\subsection{Lipschitz-continuity and Lipschitz Constants}

For a Lipschitz-continuous mapping $f:\mathbb{F}^N\to\mathbb{F}^M$ satisfying Eq.~\eqref{eq:lipschitz_continuity}, its Lipschitz constant is defined by
\begin{equation}
\operatorname{Lip}(f) = \sup_{\mathbf{x}\neq\mathbf{y}\in\mathbb{F}^N}\frac{\|f(\mathbf{x})
- f(\mathbf{y})\|_2}{\|\mathbf{x} - \mathbf{y}\|_2}.
\end{equation}
Let $f:\mathbb{F}^N\to\mathbb{F}^M$ be the composition of two Lipschitz-continuous mappings $h:\mathbb{F}^N\to\mathbb{F}^P$ and $g:\mathbb{F}^P\to\mathbb{F}^M$, i.e., $f = g \circ h$.
Then $f$ is also Lipschitz continuous, and its Lipschitz constant satisfies $\operatorname{Lip}(f) \leq \operatorname{Lip}(g)\cdot\operatorname{Lip}(h)$.

For a real-valued and differentiable mapping $f:\mathbb{R}^N\to\mathbb{R}^M$, its Lipschitz constant is equal to the supremum of the operator norm of its Jacobian matrix, i.e.,
\begin{equation}
\operatorname{Lip}(f) = \sup_{\mathbf{x}\in\mathbb{R}^N} \|\mathbf{J}_f(\mathbf{x})\|_{\mathrm{op}},
\label{eq:lip_jacobian}
\end{equation}
where $\|\cdot\|_\mathrm{op}$ is the operator norm, which corresponds to the maximum singular value of the matrix, and $\mathbf{J}_f(\mathbf{x})\in\mathbb{R}^{M\times N}$ is the Jacobian matrix of $f$ at $\mathbf{x}$ defined as follows:
\begin{equation}
\mathbf{J}_f(\mathbf{x}) = \begin{pmatrix}\frac{\partial f_1}{\partial x_1}(\mathbf{x})&\cdots&\frac{\partial f_1}{\partial x_N}(\mathbf{x})\\
\vdots&\ddots&\vdots\\
\frac{\partial f_M}{\partial x_1}(\mathbf{x})&\cdots&\frac{\partial f_M}{\partial x_N}(\mathbf{x})
\end{pmatrix}.
\end{equation}
For a complex-valued mapping $f:\mathbb{C}^N\to\mathbb{C}^M$, we identify it with the real-valued mapping $\widetilde{f}:\mathbb{R}^{2N}\to\mathbb{R}^{2M}$ obtained by stacking the real and imaginary parts, and apply the same characterization in Eq.~\eqref{eq:lip_jacobian}.
The Lipschitz constants of non-differentiable functions (e.g., ReLU) are characterized similarly using Clarke's generalized subdifferentials \cite{clarke1990optimization}. 
Since Lipschitz-continuous mappings are differentiable almost everywhere by Rademacher's theorem, we use derivative-based notation for simplicity in this paper and omit explicit treatment of generalized subdifferentials.

\subsection{PnP Algorithms and Their Convergence}
Signal recovery aims to recover a signal $\mathbf{x}\in\mathbb{R}^N$ from an observed signal $\mathbf{y}\in\mathbb{R}^M$ obtained through a degradation or noisy measurement process.
Many signal recovery tasks can be formulated as the following optimization problem:
\begin{equation}
    \min_{\hat{\mathbf{x}}\in\mathbb{R}^N} \Bigl(\alpha \mathfrak{F}_\mathbf{y}(\hat{\mathbf{x}})+\mathfrak{G}(\hat{\mathbf{x}})\Bigr),
    \label{eq:general_formulation}
\end{equation}
where $\mathfrak{F}_\mathbf{y}:\mathbb{R}^N\to\mathbb{R}\cup\{+\infty\}$ is a data fidelity function and $\alpha>0$ is its weight.
$\mathfrak{G}:\mathbb{R}^N\to\mathbb{R}\cup\{+\infty\}$ is a regularization function that models prior knowledge on the signal.
For this regularization, hand-crafted regularization functions (e.g., the $\ell_1$-norm) have been utilized \cite{boydProximalAlgorithms}.

PnP is a framework that applies a denoiser $\mathcal{D}:\mathbb{R}^N\to\mathbb{R}^N$ in place of the regularization \cite{ryuPlugandPlayMethodsProvably2019,pesquetLearningMaximallyMonotone2021,huraultGRADIENTSTEPDENOISER2022,huraultProximalDenoiserConvergent,yukawaMonotoneLipschitzGradientDenoiser2025}.
An example of a PnP method is the following ADMM-based algorithm \cite{ryuPlugandPlayMethodsProvably2019}:
\begin{equation}
\label{eq:plug-and-play}
\text{PnP-ADMM:}\;\left\lfloor\quad
\begin{aligned}
\hat{\mathbf{x}}^{[k+1]} &= \mathcal{D}(\boldsymbol{\xi}^{[k]}-\boldsymbol{\upsilon}^{[k]}),\\
\boldsymbol{\xi}^{[k+1]} &= \mathrm{prox}_{\alpha \mathfrak{F}_\mathbf{y}}(\hat{\mathbf{x}}^{[k+1]}+\boldsymbol{\upsilon}^{[k]}),\\
\boldsymbol{\upsilon}^{[k+1]} &= \boldsymbol{\upsilon}^{[k]}+\hat{\mathbf{x}}^{[k+1]}-\boldsymbol{\xi}^{[k+1]},
\end{aligned}
\right.
\end{equation}
where $k\in\mathbb{N}$ is the iteration counter, $\boldsymbol{\xi},\boldsymbol{\upsilon}\in\mathbb{R}^N$, and the proximity operator is defined as
\begin{equation}
    \mathrm{prox}_{\alpha\mathfrak{F}_\mathbf{y}}(\mathbf{v}) = \underset{\mathbf{u}\in\mathbb{R}^N}{\mathrm{argmin}}\;\left(\alpha\mathfrak{F}_\mathbf{y}(\mathbf{u}) + \frac{1}{2}\|\mathbf{u}-\mathbf{v}\|_2^2\right).
    \label{eq:prox}
\end{equation}
The algorithm in Eq.~\eqref{eq:plug-and-play} replaces the proximity operator $\mathrm{prox}_{\mathfrak{G}}:\mathbb{R}^N\to\mathbb{R}^N$ in the standard ADMM with a denoiser $\mathcal{D}$ (typically a DNN pre-trained on a denoising task) to take advantage of data-driven regularization.

To guarantee the convergence of PnP algorithms, the denoiser $\mathcal{D}$ must satisfy certain conditions, e.g., 
firm nonexpansiveness (i.e., $\|\mathcal{D}(\mathbf{x})-\mathcal{D}(\mathbf{y})\|^2\leq\langle\mathbf{x}-\mathbf{y}|\mathcal{D}(\mathbf{x})-\mathcal{D}(\mathbf{y})\rangle$ for all $\mathbf{x},\mathbf{y}\in\mathbb{R}^N$, where $\langle\cdot|\cdot\rangle$ denotes the inner product) \cite{pesquetLearningMaximallyMonotone2021},
contractivity of the residual map (i.e., $\operatorname{Lip}(\mathcal{R})<1$, where $\mathcal{R}=\mathrm{Id}-\mathcal{D}$) \cite{ryuPlugandPlayMethodsProvably2019}, or 
the gradient-step property (i.e., $\mathcal{D}=\mathrm{Id}-\nabla \mathfrak{G}$ for some $\mathfrak{G}:\mathbb{R}^N\to\mathbb{R}$) \cite{huraultGRADIENTSTEPDENOISER2022,huraultProximalDenoiserConvergent,yukawaMonotoneLipschitzGradientDenoiser2025}.
In these approaches, the construction of a Lipschitz-continuous DNN is essential, while most of them additionally require estimation of their Lipschitz constants to guarantee convergence.

\section{LipsAM (Lipschitz-continuous Amplitude Modifier): Definition and Architectures}
\label{sec:LipsAM}

This section proposes AM-based DNN architectures with guaranteed Lipschitz continuity (Ingredient~A in Fig.~\ref{fig:diagram}) to realize a convergent PnP algorithm.
The proposed architecture, termed \textit{LipsAM}, is defined in Definition~\ref{def:lipsam} and illustrated in Table~\ref{fig:architectures}.
For the LipsAM introduced here, we prove its Lipschitz continuity (Theorem~\ref{thm:lipsam}), provide a characterization (Theorem~\ref{prop:zero_preserving}), and propose a construction method (Theorem~\ref{thm:mask_sufficient}).

\begin{table*}
    \renewcommand{\arraystretch}{1.0}
    \footnotesize
    \setlength{\tabcolsep}{2pt}
    \caption{DNN Architectures of AMs (Amplitude Modifiers) and LipsAMs (Lipschitz-Continuous AMs)\textsuperscript{*}}
    \begin{tabularx}{\textwidth}{l*{7}{C}}
        \toprule
        \multirow{2}{*}{Type}& \multicolumn{3}{c}{Estimator-based} & \multicolumn{3}{c}{Masking-based} & 
        \multirow{2}{*}{General}\\
        \cmidrule(lr){2-4}\cmidrule(lr){5-7}
         & AM-E & LipsAM-E& 
        ReM-LipsAM-E& 
        AM-M & LipsAM-M & ReM-LipsAM-M &\\
        \midrule
         &
        \includegraphics[width=0.115\textwidth]{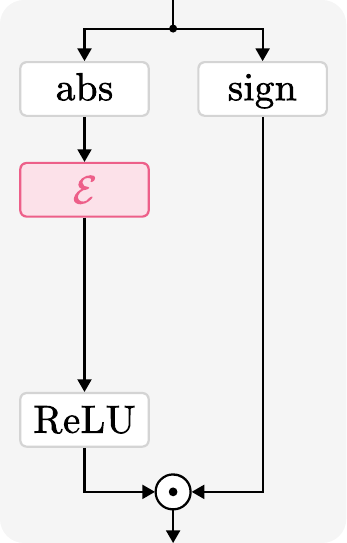} &
        \includegraphics[width=0.115\textwidth]{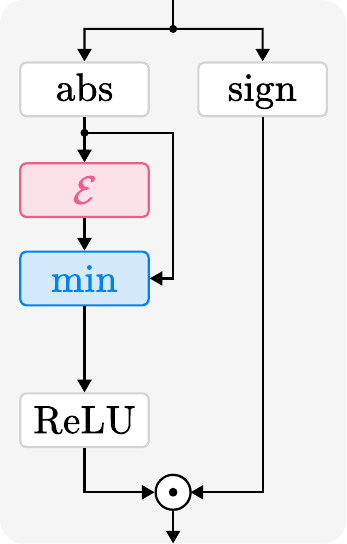} &
        \includegraphics[width=0.115\textwidth]{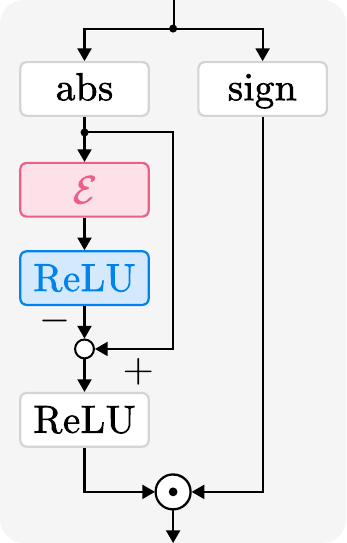} &
        \includegraphics[width=0.115\textwidth]{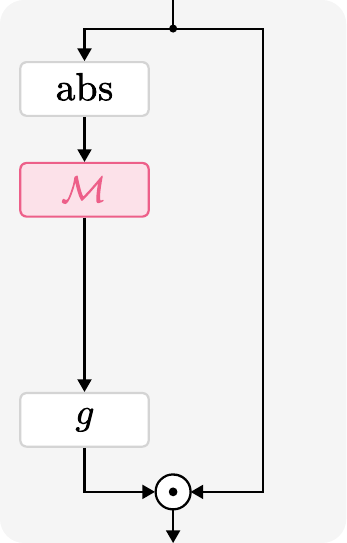} &
        \includegraphics[width=0.115\textwidth]{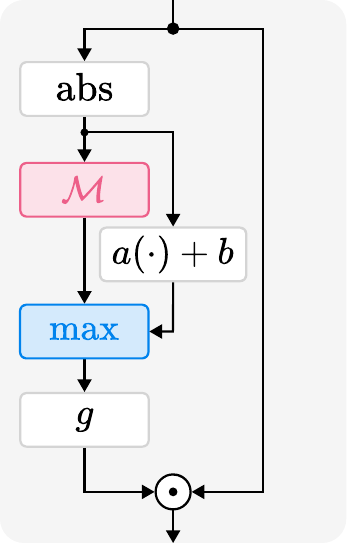} &
        \includegraphics[width=0.115\textwidth]{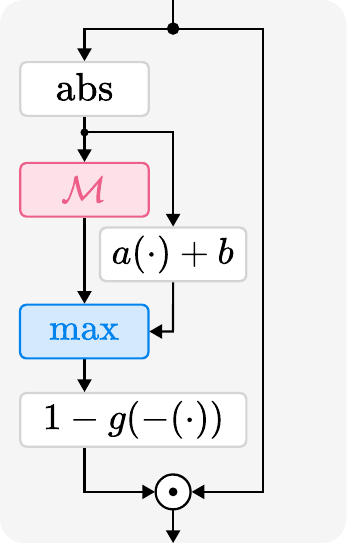} &
        \includegraphics[width=0.115\textwidth]{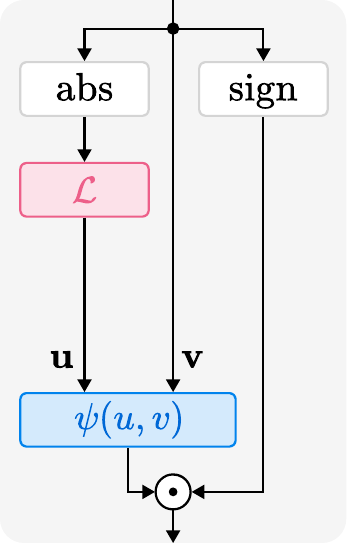} \\
        \midrule
        Notation & $\mathcal{D}_\mathcal{E}$& 
                   $\mathcal{D}_\mathcal{E}^\mathrm{(Lips)}$&                   
                   $\mathcal{D}_\mathcal{E}^\mathrm{(\mathrm{ReM\text{-}Lips})}$&                   
                   $\mathcal{D}_{(\mathcal{M},g)}$ & 
                   $\mathcal{D}_{(\mathcal{M},g,a,b)}^\mathrm{(Lips)}$&
                   $\mathcal{D}_{(\mathcal{M},g,a,b)}^\mathrm{(ReM\text{-}Lips)}$&
                   $\mathcal{D}_{(\mathcal{L},\psi)}$\\
        \midrule
        Lipschitz? & No & Yes & Yes & No & Yes (Theorem \ref{thm:lipsam_sm}) & Yes & Yes (Theorem \ref{thm:generalLipsAM})\\
        \midrule
        Bound & --- & $\sqrt{(\operatorname{Lip}(\mathcal{E}))^2+1}$ & $\operatorname{Lip}(\mathcal{E})+1$ & 
        ---  & 
        \multicolumn{2}{c}{Numerical computation (Sect.~\ref{sec:lips_bound_se_re})} & Theorem~\ref{thm:worst_case_lipschitz_bound}\\
        \bottomrule
    \end{tabularx}
    \label{fig:architectures}
    \\[4pt]
    \footnotesize{
    \textsuperscript{*}%
    The blue-shaded layers in the proposed LipsAMs guarantee the Lipschitz continuity of the overall network, whereas conventional AMs are generally not Lipschitz continuous (Sect.~\ref{sec:LipsAM}).
    The estimator-based LipsAMs were introduced in our preliminary conference publication \cite{matsumotoLIPSAMLIPSCHITZCONTINUOUSAMPLITUDE}, whereas the masking-based LipsAMs are newly proposed in this paper.
    The general form in the rightmost column is utilized for establishing an efficient framework for evaluating Lipschitz constants (Sect.~\ref{sec:framework}), and the tightest global upper bounds on $\operatorname{Lip}(\mathcal{D})$ for each architecture (Sect.~\ref{sec:lips_bound_se_re}) are summarized in the bottom row.
    Symbols of the form $\mathcal{D}:\mathbb{C}^N\to\mathbb{C}^N$ denote full DNN mappings, whereas $\mathcal{E}$, $\mathcal{M}$, and $\mathcal{L}:\mathbb{R}_+^N\to\mathbb{R}^N$ represent internal learnable operators. The function $g$ is bounded (e.g., the sigmoid function), $a>0$, $b\in\mathbb{R}$, and $\psi:\mathbb{R}\times\mathbb{R}_+\to\mathbb{R}_+$.
    }
\end{table*}

\subsection{An AM in Eq.~\eqref{eq:amplitude_modifier} is not Lipschitz Continuous in General}
\label{sec:pitfall}

An AM modifies the amplitude components of a complex-valued signal while preserving its phase components as in Eq.~\eqref{eq:amplitude_modifier}.
A key issue with this approach is that the Lipschitz continuity of the amplitude-modifying part $\mathcal{A}$ alone does not imply the Lipschitz continuity of the overall mapping $\mathcal{D}_{\!\mathcal{A}}$.
This issue stems from the discontinuity of $\operatorname{sign}(z_n)$ at $z_n=0$ for some $n\in[N]$, as demonstrated in the following examples.

\begin{example}[Bias]
\label{ex:bias}
Let $\mathcal{A}:\mathbb{R}_+\to\mathbb{R}_+$ be given by $\mathcal{A}(x)=x+1$.
Then $\mathcal{A}$ is Lipschitz continuous.
However, the corresponding AM given by $\mathcal{D}_{\!\mathcal{A}}(z)=(|z|+1)\cdot\operatorname{sign}(z)$ is not even continuous at $z=0$ and is therefore not Lipschitz continuous.
\end{example}

\begin{example}[Permutation]
\label{ex:permutation}
Let $\mathcal{A}:\mathbb{R}_+^2\to\mathbb{R}_+^2$ be given by $\mathcal{A}(x_1,x_2)=(x_2,x_1)$.
Then $\mathcal{A}$ is Lipschitz continuous.
However, the corresponding AM given by $\mathcal{D}_{\!\mathcal{A}}(z_1,z_2) =(|z_2|\cdot\operatorname{sign}(z_1),|z_1|\cdot\operatorname{sign}(z_2))$ is not Lipschitz continuous.
Indeed, $(\mathcal{D}_{\!\mathcal{A}}(z_1,1))_1=\operatorname{sign}(z_1)$ is not continuous at $z_1=0$.
\end{example}

\subsection{Characterization of Lipschitz-continuous AMs}

We now clarify the conditions on the amplitude-modifying parts to guarantee the Lipschitz continuity of the AMs.
To this end, we propose a class of AMs named \textbf{LipsAM} as follows:
\begin{dfn}[LipsAM]
\label{def:lipsam}
A mapping $\mathcal{D}_{\!\mathcal{A}}:\mathbb{C}^N\to\mathbb{C}^N$ is a \textit{LipsAM} if there exists $\mathcal A:\mathbb R_+^N\to\mathbb R_+^N$ such that $\mathcal D_{\!\mathcal A}(\mathbf z) = \mathcal A(|\mathbf z|)\odot\operatorname{sign}(\mathbf z)$ for all $\mathbf z\in\mathbb C^N$, and the following two conditions are satisfied for some $L_1,L_2\ge0$.
\begin{enumerate}[label=(\roman*)]
    \item $\mathcal{A}$ is $L_1$-Lipschitz continuous, i.e., for all $\mathbf{x},\mathbf{y}\in\mathbb{R}_+^N$,
    \begin{equation}
        \|\mathcal{A}(\mathbf{x}) - \mathcal{A}(\mathbf{y}) \|_2 \leq L_1 \|\mathbf{x}-\mathbf{y}\|_2.
        \label{eq:cond1}
    \end{equation}
    
    \item For all $\mathbf{x}\in\mathbb{R}_+^N$ and all $n\in[N]$,
    \begin{equation}
    (\mathcal{A}(\mathbf{x}))_n \le L_2\,x_n.
    \label{eq:cond2}
    \end{equation}
\end{enumerate}
\end{dfn}

We show that LipsAMs are equivalent to the set of all Lipschitz-continuous AMs, which is a key result of this work%
\footnote{
In our conference paper \cite{matsumotoLIPSAMLIPSCHITZCONTINUOUSAMPLITUDE}, the conditions of LipsAM in Definition~\ref{def:lipsam} were shown to be sufficient for
an AM to be Lipschitz continuous.
Theorem~\ref{thm:lipsam} in this paper further establishes their necessity, completing the equivalence.
}.

\begin{thm}[Characterization of Lipschitz-continuous AMs]
\label{thm:lipsam}
Let $\mathcal{D}_{\!\mathcal{A}}:\mathbb{C}^N\to\mathbb{C}^N$ be an AM.
Then $\mathcal{D}_{\!\mathcal{A}}$ is Lipschitz continuous if and only if $\mathcal{D}_{\!\mathcal{A}}$ is a LipsAM in Definition~\ref{def:lipsam}.
\begin{proof}
See Appendix~\ref{app:proof_thm_lipsam}.
\end{proof}
\end{thm}

We next reformulate the second condition into a more tractable form for designing LipsAM architectures.

\begin{thm}[Zero-preserving characterization]
\label{prop:zero_preserving}
Let $\mathcal{D}_{\!\mathcal{A}}:\mathbb{C}^N\to\mathbb{C}^N$ be an AM given by $\mathcal{D}_{\!\mathcal{A}}(\mathbf{z})=\mathcal{A}(|\mathbf{z}|)\odot\operatorname{sign}(\mathbf{z})$, where $\mathcal{A}:\mathbb{R}_+^N\to\mathbb{R}_+^N$.
Then $\mathcal{D}_{\!\mathcal{A}}$ is a LipsAM (equivalently, Lipschitz continuous) if and only if $\mathcal{A}$ is Lipschitz continuous and, for any $n\in[N]$ and $\mathbf{x}\in\mathbb{R}_+^N$, 
\begin{equation}
    x_n=0\;\Rightarrow\;(\mathcal{A}(\mathbf{x}))_n=0.
    \label{eq:zero_preserving}
\end{equation}
\end{thm}
\begin{proof}
See Appendix~\ref{app:proof_prop_zero_preserving}.
\end{proof}

We refer to Eq.~\eqref{eq:zero_preserving} as the \textit{zero-preserving property}.
This property prevents discontinuous behavior of AMs, such as those seen in Examples~\ref{ex:bias} and \ref{ex:permutation}.
Moreover, the zero-preserving property can be easily verified from the structure of a DNN.

\subsection{LipsAM-E: Direct Estimation of Amplitude via LipsAM}

To show the usefulness of the zero-preserving characterization, let us introduce AM-E (``-E'' denotes ``estimator''), an AM that directly estimates the amplitude components using a learnable mapping $\mathcal{E}:\mathbb{R}_+^N \to \mathbb{R}^N$, defined as 
\begin{equation}
\text{AM-E :}\quad\mathcal{D}_{\mathcal{E}}(\mathbf{z}) = (\mathcal{E}(|\mathbf{z}|))_+ \odot \operatorname{sign}(\mathbf{z}),
\label{eq:signal}
\end{equation}
where $(\cdot)_+=\max(\cdot,0)$ denotes the ReLU function.
Note that, even if $\mathcal{E}$ is assumed to be Lipschitz continuous, $\mathcal{D}_{\mathcal{E}}$ is generally not Lipschitz continuous since the zero-preserving property in Eq.~\eqref{eq:zero_preserving} does not necessarily hold.

To guarantee Lipschitz continuity, we propose \textbf{LipsAM-E} by incorporating the element-wise $\min$ layer (highlighted in blue in Table~\ref{fig:architectures}) to AM-E as follows:
\begin{align}
\text{LipsAM-E : }\;&\mathcal{D}_{\mathcal{E}}^\mathrm{(Lips)}
(\mathbf{z}) = \nonumber\\
&\qquad(\min(\mathcal{E}(|\mathbf{z}|), |\mathbf{z}|))_+ \odot \operatorname{sign}(\mathbf{z}).
\label{eq:signallim}
\end{align}

The Lipschitz continuity of LipsAM-E is guaranteed whenever $\mathcal{E}$ is Lipschitz continuous, which can be verified through the zero-preserving property. 
The amplitude-modifying part of LipsAM-E is written as follows:
\begin{align}
\mathcal{A}_{\mathcal{E}}^\mathrm{(Lips)}(\mathbf{x}) &= (\min(\mathcal{E}(\mathbf{x}), \mathbf{x}))_+.
\label{eq:AinSE}
\end{align}
If $\mathcal{E}$ is Lipschitz continuous, then the mapping $\mathcal{A}_{\mathcal{E}}^\mathrm{(Lips)}$ is also Lipschitz continuous.
Furthermore, for any input with 0 elements (i.e., $x_n = 0$), the output $(\mathcal{A}_{\mathcal{E}}^\mathrm{(Lips)}(\mathbf{x}))_n = (\min((\mathcal{E}(\mathbf{x}))_n, 0))_+$ is always zero.
Therefore, $\mathcal{A}_{\mathcal{E}}^\mathrm{(Lips)}$ satisfies the zero-preserving property, which implies that $\mathcal{D}_{\mathcal{E}}^\mathrm{(Lips)}$ is a LipsAM (and thus Lipschitz continuous) by Theorem~\ref{prop:zero_preserving}.

\subsection{LipsAM-M: Lipschitz-continuous Time-Frequency Masking}
\label{sec:masking}

Next, we introduce AM-M defined as follows:
\begin{align}
\text{AM-M}:\quad \mathcal{D}_{(\mathcal{M},g)}(\mathbf{z}) &= G(\mathcal{M}(|\mathbf{z}|)) \odot \mathbf{z}
\nonumber
\\
&=(G(\mathcal{M}(|\mathbf{z}|)) \odot |\mathbf{z}|) \odot \operatorname{sign}(\mathbf{z})\label{eq:masking},
\end{align}
where $\mathcal{M}:\mathbb{R}_+^N \to \mathbb{R}^N$ is a learnable mapping, $g:\mathbb{R}\to[0,B_g]$ is a bounded non-negative function with upper bound $B_g\geq0$, and $G:\mathbb{R}^N\to[0,B_g]^N$ applies $g$ element-wise as $G(\mathbf{t})=(g(t_n))_{n=1}^N$ for all $\mathbf{t}\in\mathbb{R}^N$.
AM-M is based on time-frequency masking, a typical technique in audio signal processing \cite{narayananIdealRatioMask2013,hersheyDeepClusteringDiscriminative2016} (``-M'' denotes masking).
The sigmoid function is often used for $g$, and in this case $B_g=1$.

Since an AM-M is generally not a LipsAM, we derive a sufficient condition for the Lipschitz continuity of AM-M (Theorem~\ref{thm:mask_sufficient}) and propose LipsAM-M in Eq.~\eqref{eq:LipsAM-M}.

\subsubsection{A Masking-based AM is not a LipsAM in General}
The amplitude-modifying part of AM-M in Eq.~\eqref{eq:masking} is
\begin{equation}
    \mathcal{A}_{(\mathcal{M},g)}(\mathbf{x}) = G(\mathcal{M}(\mathbf{x})) \odot \mathbf{x}.
\end{equation}
Although this mapping admits the zero-preserving property (i.e., $(\mathcal{A}_{(\mathcal{M},g)}(\mathbf{x}))_n = 0$ if $x_n=0$), it is not Lipschitz continuous in general.
\begin{example}
\label{ex:non-lipschitz}
Let $\mathcal{M}(x)=\sin(x)$ and $g(t)=\mathrm{sigmoid}(4t)$.
Then, both $\mathcal{M}$ and $g$ are $1$-Lipschitz continuous%
\footnote{
The factor $4$ inside $g(t)=\mathrm{sigmoid}(4t)$ is chosen to make Fig.~\ref{fig:nonLipsExample} visually clear, while also ensuring $\operatorname{Lip}(g)=1$.
This factor is not essential to provide a non-Lipschitz example. 
For instance, $\mathcal{A}_{(\mathcal{M},g)}$ is still not Lipschitz continuous for $\mathcal{M}(x)=\sin(x)$ and $g(t)=\mathrm{sigmoid}(t)$.
}.
However, the mapping $\mathcal{A}_{(\mathcal{M},g)}(x)=\mathrm{sigmoid}(4\sin(x))\cdot x$ is not Lipschitz continuous because $(\mathrm{d}\mathcal{A}_{(\mathcal{M},g)}/\mathrm{d}x)=
\mathrm{sigmoid}(4\sin(x))
+
4x\cos(x)\,
\mathrm{sigmoid}(4\sin(x))
(1-\mathrm{sigmoid}(4\sin(x)))
$ can be arbitrarily large as $x\to+\infty$
(see the gray lines in Fig.~\ref{fig:nonLipsExample}).
\end{example}

We now analyze why $\mathcal{A}_{(\mathcal{M},g)}$ is generally not Lipschitz continuous. 
The Jacobian of $\mathcal{A}_{(\mathcal{M},g)}$ is given by 
\begin{align}
\mathbf{J}_{\mathcal{A}_{(\mathcal{M},g)}}(\mathbf{x}) = \mathrm{diag}(G(\mathcal{M}(\mathbf{x}))) + \boldsymbol{\Phi}_{(\mathcal{M},\,g)}(\mathbf{x})\,\mathbf{J}_{\mathcal{M}}(\mathbf{x}),
\label{eq:jacobian}
\end{align}
where $\boldsymbol{\Phi}_{(\mathcal{M},g)}\in\mathbb{R}^{N\times N}$ is a diagonal matrix defined by
\begin{align}
    &\boldsymbol{\Phi}_{(\mathcal{M},\,g)}(\mathbf{x})=\mathrm{diag}(\mathbf{x})\,\mathbf{J}_{G}(\mathcal{M}(\mathbf{x}))\label{eq:theopnorm}\\ 
    &\!{}={}\mathrm{diag}\left(x_1\!\cdot\!(\mathrm{d}g/\mathrm{d}t)((\mathcal{M}(\mathbf{x}))_1), \ldots, x_N\!\cdot\!(\mathrm{d}g/\mathrm{d}t)((\mathcal{M}(\mathbf{x}))_N)\right),\nonumber
\end{align}
and its operator norm is given by
\begin{align}
\|\boldsymbol{\Phi}_{(\mathcal{M},g)}(\mathbf{x})\|_{\mathrm{op}}
=\max_{n=1,\ldots,N} |x_n\cdot(\mathrm{d}g/\mathrm{d}t)((\mathcal{M}(\mathbf{x}))_n)|.
\label{eq:Phiop}
\end{align}
The Lipschitz continuity of $\mathcal{A}_{(\mathcal{M},g)}$ requires the operator norm of $\mathbf{J}_{\mathcal{A}_{(\mathcal{M},g)}}(\mathbf{x})$ in Eq.~\eqref{eq:jacobian} to be globally bounded.
However, $\|\boldsymbol{\Phi}_{(\mathcal{M},g)}(\mathbf{x})\|_{\mathrm{op}}$ is not necessarily bounded because $x_n$ can be arbitrarily large.
Note that restricting the analysis to bounded inputs is insufficient for our purpose, as an iterative algorithm without guaranteed convergence might generate an unbounded sequence.
Thus, we need some mechanism that keeps this operator norm bounded even for arbitrarily large inputs.

\begin{rem}
Clipping the input amplitude to a bounded range before applying the mask can ensure Lipschitz continuity, yet this modification does not preserve the form in Eq.~\eqref{eq:masking}.
\end{rem}

\begin{figure}
\centering
\includegraphics[scale=0.5]{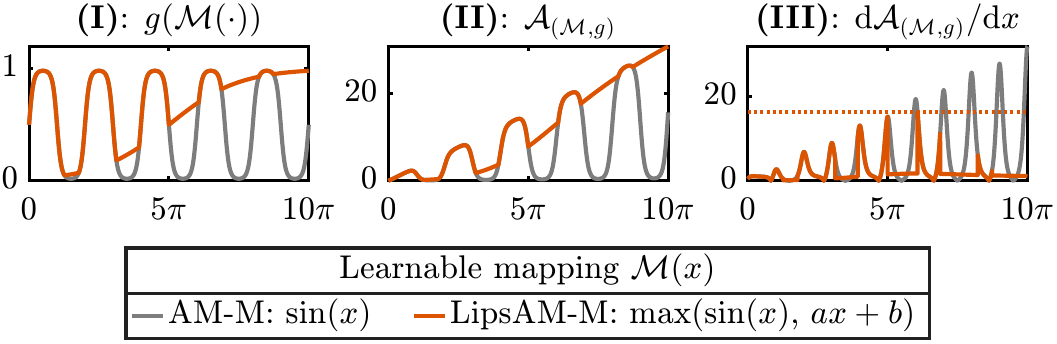}
\caption{
    Visualization of Example~\ref{ex:non-lipschitz} (non-Lipschitz continuous AM-M).
    The mask value $g(\mathcal{M}(\cdot))$ of this AM-M shown in \textbf{(I)} continues to oscillate even for large inputs,
    causing the amplitude-modifying part $\mathcal{A}_{(\mathcal{M},g)}(x)=g(\mathcal{M}(x))\cdot x$ in \textbf{(II)} to be non-Lipschitz continuous.
    The proposed LipsAM-M (red) resolves this issue; 
    the derivative of $\mathcal{A}_{(\mathcal{M},g)}$ is bounded by the dotted horizontal line in \textbf{(III)}, thus $\mathcal{A}_{(\mathcal{M},g)}$ is Lipschitz continuous.
    Here, we set $g(t)=\mathrm{sigmoid}(4t)$, $a=1/5\pi$, and $b=-1$.
}
\label{fig:nonLipsExample}
\end{figure}

\subsubsection{A Sufficient Condition for Lipschitz-continuity of Masking Operators}

To make AM-M (i.e., $\mathcal{D}_{(\mathcal{M},g)}$ in Eq.~\eqref{eq:masking}) Lipschitz continuous, we derive a sufficient condition on $\mathcal{M}$ and $g$ for global boundedness of $\|\boldsymbol{\Phi}_{(\mathcal{M},g)}(\mathbf{x})\|_{\mathrm{op}}$ in Eq.~\eqref{eq:Phiop}.

\begin{thm}[Sufficient condition for Lipschitz-continuity]
\label{thm:mask_sufficient}
Let $\mathcal{M}:\mathbb{R}_+^N \to \mathbb{R}^N$ and $g:\mathbb{R} \to [0, B_g]$ be Lipschitz continuous.
Define $G:\mathbb{R}^N\to[0, B_g]^N:\mathbf{t}\mapsto(g(t_n))_{n=1}^N$.
Then $\mathcal{D}_{(\mathcal{M},g)}(\mathbf{z})=G(\mathcal{M}(|\mathbf{z}|)) \odot \mathbf{z}$ is a LipsAM if both of the following conditions hold:
\begin{enumerate}[label=(\roman*)]
    \item There exists a constant $C_\mathcal{M}\geq 0$ independent of $n$, such that, 
    for every sequence $(\mathbf{x}^{[k]})_{k\in\mathbb{N}}$ in $\mathbb{R}_+^N$ and $n\in[N]$, 
    \begin{equation}
    x_n^{[k]} \to +\infty\;\Rightarrow\;\limsup_{k \to \infty} \left| \frac{x_n^{[k]}}{(\mathcal{M}(\mathbf{x}^{[k]}))_n} \right| \leq C_\mathcal{M}.
    \label{eq:condition1}
    \end{equation}
    \item $|t\cdot (\mathrm{d}g/\mathrm{d}t)(t)|$ is globally bounded on $t\in\mathbb{R}$.
\end{enumerate}
\end{thm}
\begin{proof}
    See Appendix~\ref{sec:proof_thm:mask_sufficient}.
\end{proof}

Condition~(i) requires that when $x_n$ diverges, $|(\mathcal{M}(\mathbf{x}))_n|$ must grow at least linearly.
Condition~(ii) requires that $(\mathrm{d}g/\mathrm{d}t)(t)$ decays to zero at a rate of at least $1/|t|$.
Together, these conditions ensure that $(\mathrm{d}g/\mathrm{d}t)((\mathcal{M}(\mathbf{x}))_n)$ decays sufficiently fast relative to the growth of $x_n$, thereby keeping $|x_n\cdot(\mathrm{d}g/\mathrm{d}t)((\mathcal{M}(\mathbf{x}))_n)|$ globally bounded.

Note that Condition~(ii) is mild,
as typical functions for $g$, including the sigmoid function, satisfy this condition. 
Thus, the remaining requirement is the construction of $\mathcal{M}$ satisfying Condition~(i), which we address in the next subsubsection.

\subsubsection{LipsAM-M: A Lipschitz-continuous Masking-based Architecture}

We propose \textbf{LipsAM-M}, a globally Lipschitz-continuous masking-based AM, which we define as follows:
\begin{align}
&\text{LipsAM-M}:{}
\mathcal{D}_{(\mathcal{M},g,a,b)}^\mathrm{(Lips)}(\mathbf{z})
= \nonumber \\
&\qquad\qquad\qquad\qquad G\Bigl(\max(\mathcal{M}(|\mathbf{z}|),\, a|\mathbf{z}| + b \mathbf{1})\Bigr)\odot \mathbf{z},
\label{eq:LipsAM-M}
\end{align}
where $a>0$ and $b\in\mathbb{R}$ are parameters.
This architecture is illustrated in Table~\ref{fig:architectures}.
Below, we show the Lipschitz continuity of this architecture using Theorem~\ref{thm:mask_sufficient}.

\begin{thm}
\label{thm:lipsam_sm}
Assume that $\mathcal{M}:\mathbb{R}_+^N \to \mathbb{R}^N$ and $g:\mathbb{R}\to[0,B_g]$ are Lipschitz continuous, and $g$ satisfies Condition~(ii) in Theorem~\ref{thm:mask_sufficient}.
Let $a>0$ and $b\in\mathbb{R}$.
Then $\mathcal{D}_{(\mathcal{M},g,a,b)}^\mathrm{(Lips)}$ defined in Eq.~\eqref{eq:LipsAM-M} is a LipsAM.
\end{thm}
\begin{proof}
See Appendix~\ref{app:proof_thm_lipsam_sm}.
\end{proof}

The key to the proof is that the max operation in Eq.~\eqref{eq:LipsAM-M} enforces Condition~(i) in Theorem~\ref{thm:mask_sufficient}, regardless of the choice of $\mathcal{M}$.
As an example, Fig.~\ref{fig:nonLipsExample} compares LipsAM-M with the AM-M in Example~\ref{ex:non-lipschitz}.
By replacing $\mathcal{M}(x)=\sin(x)$ with $\max(\sin(x), ax+b)$, the derivative of the amplitude-modifying part is successfully bounded, as shown in Fig.~\ref{fig:nonLipsExample}~(III).

\begin{rem}
LipsAM-M in Eq.~\eqref{eq:LipsAM-M} converges pointwise to AM-M in Eq.~\eqref{eq:masking} as $b\to-\infty$ with $a>0$ remaining bounded.
\end{rem}

\subsection{Using LipsAMs as Noise Estimator}
\label{sec:residualmap}

A denoiser in residual form, $\mathcal{D}=\mathrm{Id}-\mathcal{R}$, tends to be easier to train.
We therefore consider using LipsAMs as $\mathcal{R}$.

By plugging LipsAM-E into the residual form, we propose \textbf{ReM-LipsAM-E} defined as follows:
\begin{equation}
\text{ReM-LipsAM-E :}\quad \mathcal{D}_{\mathcal{E}}^{\mathrm{(ReM\text{-}Lips)}}(\mathbf{z}) = \mathbf{z} - \mathcal{D}_{\mathcal{E}}^\mathrm{(Lips)}(\mathbf{z}),
\label{eq:ReMLipsAME}
\end{equation}
where ``ReM-'' denotes residual mapping.
It admits another formulation by rewriting its amplitude-modifying part as
\begin{align}
    \mathcal{A}_{\mathcal{E}}^{\mathrm{(ReM\text{-}Lips)}}(\mathbf{x}) & = (\mathbf{x} - (\mathcal{E}(\mathbf{x}))_+)_+.
    \label{eq:remimplement}
\end{align}
That is, ReM-LipsAM-E can be implemented by applying ReLU to the estimated residual component $\mathcal{E}(\mathbf{x})$.
The illustration in Table~\ref{fig:architectures} is based on the formulation in Eq.~\eqref{eq:remimplement}.

To use LipsAM-M in a residual form, we propose \textbf{ReM-LipsAM-M} defined as follows:
\begin{align}
&\text{ReM-LipsAM-M}:{}
\mathcal{D}_{(\mathcal{M},g,a,b)}^{\mathrm{(ReM\text{-}Lips)}}(\mathbf{z})
= \nonumber \\
&\qquad\qquad
\mathbf{z}
-\underset{\mathcal{R}(\mathbf{z})}{\underbrace{
G\Bigl(-\max(\mathcal{M}(|\mathbf{z}|),\, a|\mathbf{z}| + b \mathbf{1})\Bigr)\odot \mathbf{z}
}}.
\label{eq:R-LipsAM-M}
\end{align}
The reason for the sign inversion before applying $g$ is that a commonly used $g$ is monotonically increasing.
Without the sign inversion, large values of $|z_n|$ lead to large mask values through the path $a|z_n|+b$.
This is undesirable because the noise, which we want to estimate via $\mathcal{R}$, is typically smaller than the desired signal.
The sign inversion reverses this behavior, making the mask smaller for large elements and larger for small ones.

\begin{rem}
\label{rem:maskcoincides}
When $g$ satisfies $1-g(-t)=g(t)$, e.g., the sigmoid function, ReM-LipsAM-M coincides with LipsAM-M.
\end{rem}

\section{Evaluation of Lipschitz Constants of LipsAMs}
\label{sec:bounds}

This section provides tools for evaluating the Lipschitz constants of LipsAM architectures (Ingredient~B in Fig.~\ref{fig:diagram}).

The following proposition shows that the Lipschitz constant of a LipsAM $\mathcal{D}_{\!\mathcal{A}}$ coincides with that of its amplitude-modifying part $\mathcal{A}$, and the subsequent analysis accordingly focuses on the evaluation of $\operatorname{Lip}(\mathcal{A})$ for each architecture.
\begin{prop}
\label{prop:lip_equal}
Let $\mathcal{D}_{\!\mathcal{A}}:\mathbb{C}^N\to\mathbb{C}^N$ be a LipsAM given by
$\mathcal{D}_{\!\mathcal{A}}(\mathbf{z})
=\mathcal{A}(|\mathbf{z}|)\odot\operatorname{sign}(\mathbf{z})$,
where $\mathcal{A}:\mathbb{R}_+^N\to\mathbb{R}_+^N$.
Then,
\begin{equation}
\operatorname{Lip}(\mathcal{D}_{\!\mathcal{A}})
=
\operatorname{Lip}(\mathcal{A}).
\end{equation}
\end{prop}
\begin{proof}
See Appendix~\ref{app:proof_prop_lip_equal}.
\end{proof}

As in Eq.~\eqref{eq:lip_jacobian}, computing a Lipschitz constant requires the evaluation of the supremum of the operator norm of its Jacobian (strictly speaking, Clarke's generalized Jacobian) over a potentially high-dimensional search space.
To make this computation tractable, we propose a two-step framework for efficiently evaluating the Lipschitz constants of arbitrary LipsAM architectures as follows:
\begin{enumerate}
    \item Represent a LipsAM in the form of Eq.~\eqref{eq:generalLipsAM}, where
    an bivariate function $\psi$ characterizes an architecture.
    \item Solve the seven-variable optimization problem associated with the function $\psi$ in Eq.~\eqref{eq:worst_case_lipschitz_bound_2d}.
\end{enumerate}
The second step is enabled by Theorem~\ref{thm:worst_case_lipschitz_bound}, which reduces the search space for evaluating the Lipschitz constant to a lower-dimensional one.
Then, we evaluate the Lipschitz constants of the proposed LipsAM architectures, deriving closed-form formulas for the estimator-based architectures and numerically computing them for the masking-based architectures.

\subsection{General Representation of AM Architectures}
\label{sec:generalRep}

We introduce the following general representation that characterizes an AM architecture by a bivariate function $\psi:\mathbb{R}\times\mathbb{R}_+\to \mathbb{R}_+$, which we use for the following analysis.
\begin{equation}
    \text{General AM:}\quad
    \mathcal{D}_{(\mathcal{L},\psi)}(\mathbf{z}) = \underset{=\mathcal{A}_{(\mathcal{L},\psi)}(|\mathbf{z}|)}{\underbrace{\Psi(\mathcal{L}(|\mathbf{z}|),|\mathbf{z}|)}}\odot\operatorname{sign}(\mathbf{z}),
\label{eq:generalLipsAM}
\end{equation}
where $\mathcal{L}:\mathbb{R}_+^N\to\mathbb{R}^N$ is the learnable mapping (i.e., $\mathcal{E}$ and $\mathcal{M}$ in estimator- and masking-based AM, respectively), and $\Psi:\mathbb{R}^N\times \mathbb{R}_+^N\to \mathbb{R}_+^N$ is a function that applies the bivariate function $\psi$ element-wise, given by $\Psi(\mathbf{u},\mathbf{v}) = (\psi(u_n,v_n))_{n=1}^N$.
Specifically, LipsAM-E, ReM-LipsAM-E, and LipsAM-M can be represented in this form by setting $\psi$ as follows%
\footnote{
We omit ReM-LipsAM-M ($\psi(u,v)=(1-g(-\max(u,a v + b)))\cdot v$) from the analysis in this section since we use $g=\operatorname{sigmoid}$ in this paper, and in this case ReM-LipsAM-M coincides with LipsAM-M (see Remark~\ref{rem:maskcoincides}).
}:
\begin{align}
\text{LipsAM-E:}\quad & \psi(u,v) = (\min(u,v))_+, \\
\text{ReM-LipsAM-E:}\quad & \psi(u,v) = (v-(u)_+)_+, \label{eq:phiREME}\\
\text{LipsAM-M:}\quad & \psi(u,v) = g(\max(u,a v + b))\cdot v.
\label{eq:phiM}
\end{align}
The general forms of AMs are displayed in the rightmost column of Table~\ref{fig:architectures}. 
This general form of AM satisfies the condition of LipsAM when $\mathcal{L}$ and $\psi$ are Lipschitz continuous and the zero-preserving property for $\psi$ is satisfied, as follows.

\begin{thm}
\label{thm:generalLipsAM}
Let $\mathcal{L}:\mathbb{R}_+^N\to\mathbb{R}^N$ and $\psi:\mathbb{R}\times\mathbb{R}_+\to \mathbb{R}_+$ be Lipschitz continuous and assume that $\psi(u,0) = 0$ for all $u\in\mathbb{R}$.
Then, $\mathcal{D}_{(\mathcal{L},\psi)}$ in Eq.~\eqref{eq:generalLipsAM} is a LipsAM.
\end{thm}
\begin{proof}
$\mathcal{A}_{(\mathcal{L},\psi)}$ in Eq.~\eqref{eq:generalLipsAM} with such $\mathcal{L}$ and $\psi$ is Lipschitz continuous and preserves zeros, 
i.e., $(\mathcal{A}_{(\mathcal{L},\psi)}(\mathbf{x}))_n=0$ whenever $x_n=0$. Thus Theorem~\ref{prop:zero_preserving} applies.
\end{proof}

\subsection{Evaluating Lipschitz Constants in a Reduced Search Space}
\label{sec:framework}

Now we aim to obtain an upper bound on the Lipschitz constant for a LipsAM architecture characterized by $\psi$.
More precisely, we assume that the Lipschitz constant of the learnable mapping $\mathcal{L}:\mathbb{R}_+^N\to\mathbb{R}^N$ is upper bounded by a known value $L\geq0$, i.e., $\operatorname{Lip}(\mathcal{L})\leq L$.
We then seek the tightest upper bound on the Lipschitz constant of LipsAM,

\begin{equation}
B_{(L,\psi)}^{(N)} = \sup_{\mathcal{L}:\mathbb{R}_+^N\to\mathbb{R}^N:\operatorname{Lip}(\mathcal{L})\le L}\operatorname{Lip}\bigl(\mathcal{A}_{(\mathcal{L},\psi)}\bigr).
\label{eq:objective}
\end{equation}
A naive approach to computing this quantity is to maximize the operator norm of the Jacobian of $\mathcal{A}_{(\mathcal{L},\psi)}$, specifically,
\begin{equation}
\!\!\!B_{(L,\psi)}^{(N)} =
    \sup_{\mathcal{L}:\mathbb{R}_+^N\to\mathbb{R}^N:\operatorname{Lip}(\mathcal{L})\le L}\!
\left(\sup_{\mathbf{x}\in\mathbb{R}_+^N}
\|
\underset{\mathbf{J}_{\mathcal{A}_{(\mathcal{L},\psi)}}(\mathbf{x})}{\underbrace{\mathbf{D}_1 \mathbf{J}_{\mathcal{L}} + \mathbf{D}_2}}
\|_{\mathrm{op}}
\right)\!,\!
\label{eq:arbitrary_N}
\end{equation}
where $\mathbf{J}_{\mathcal{L}}\in\mathbb{R}^{N\times N}$ is the Jacobian matrix of $\mathcal{L}$ at $\mathbf{x}\in\mathbb{R}_+^N$, $\mathbf{D}_1\in\mathbb{R}^{N\times N}$ and $\mathbf{D}_2\in\mathbb{R}^{N\times N}$ are diagonal matrices defined as $(\mathbf{D}_1)_{nn}=\partial_{1}\psi((\mathcal{L}(\mathbf{x}))_n,x_n)$ and $(\mathbf{D}_2)_{nn}=\partial_{2}\psi((\mathcal{L}(\mathbf{x}))_n,x_n)$, and $\partial_{1}$ and $\partial_{2}$ denote the partial derivatives with respect to the first and second arguments, respectively%
\footnote{
By Rademacher's theorem, the derivative-based notations remain valid almost everywhere even when $\psi$ is not globally differentiable.
}.

However, directly evaluating Eq.~\eqref{eq:arbitrary_N} is computationally demanding, particularly when the input dimension $N\in\mathbb{N}$ is large because the Jacobian matrix is of size $N\times N$.
The dimension of the search space for $\mathcal{L}$ and $\mathbf{x}$ also rapidly increases with $N$.
To reduce this search space, we investigate whether the evaluation can be reduced to a lower-dimensional problem.
The following theorem shows that this is indeed possible, as the tightest upper bound is attained at $N=2$.

\begin{thm}[Tight global upper bound of Lipschitz constant of LipsAM]
\label{thm:worst_case_lipschitz_bound}
Let $\psi:\mathbb{R}\times\mathbb{R}_+\to\mathbb{R}_+$ be Lipschitz continuous and satisfy $\psi(\cdot,0)=0$.
Let $L\geq0$, and $B_{(L,\psi)}$ be the supremum of $B_{(L,\psi)}^{(N)}$ in Eq.~\eqref{eq:objective} over $N\in\mathbb{N}$. 
Then, this value is attained at $N=2$, i.e.,
\begin{equation}
B_{(L,\psi)}=\sup_{N\in\mathbb{N}} B_{(L,\psi)}^{(N)}
=
B_{(L,\psi)}^{(2)},
\label{eq:reduce2}
\end{equation}
and, for any $L$-Lipschitz continuous mapping $\mathcal{L}:\mathbb{R}_+^N\to\mathbb{R}^N$,
\begin{equation}
\operatorname{Lip}(\mathcal{D}_{(\mathcal{L},\psi)})
\leq
B_{(L,\psi)}.
\label{eq:tightbound}
\end{equation}
\begin{proof}
See Appendix~\ref{proof:thm_wb}.
\end{proof}
\end{thm}

\begin{rem}
As in Appendix~\ref{sec:loose}, a dimension-independent upper bound is readily given by 
$\operatorname{Lip}(\psi)\cdot\sqrt{(\operatorname{Lip}(\mathcal{L}))^2+1}\geq\operatorname{Lip}(\mathcal{D}_{(\mathcal{L},\psi)})$, 
but this bound is generally not tight.
Since the bound is used for normalization in Eq.~\eqref{eq:CRLipsAM}, its overestimation may unnecessarily restrict the expressive power of DNNs.
\end{rem}

From Theorem~\ref{thm:worst_case_lipschitz_bound}, the upper bound $B_{(L,\psi)}$ in Eq.~\eqref{eq:tightbound} can be computed via the case $N=2$.
This allows its efficient numerical estimation when a closed-form bound is not available.
For $N=2$, the diagonal matrices in Eq.~\eqref{eq:arbitrary_N} can be expressed using four parameters, $u_1,u_2\in\mathbb{R}$ and $v_1,v_2\in\mathbb{R}_+$:
\begin{align}
\mathbf{D}_1^{(u_1,u_2,v_1,v_2)}
&=
\mathrm{diag}(\partial_1\psi(u_1,v_1),\partial_1\psi(u_2,v_2)),\\
\mathbf{D}_2^{(u_1,u_2,v_1,v_2)}
&=
\mathrm{diag}(\partial_2\psi(u_1,v_1),\partial_2\psi(u_2,v_2)).
\end{align}
Moreover, since $\mathcal{L}:\mathbb{R}_+^2\to\mathbb{R}^2$ is $L$-Lipschitz continuous for some $L\geq0$, the singular values of the Jacobian matrix $\mathbf{J}_{\mathcal{L}}$ are at most $L$.
Thus, this matrix can be reparametrized as
\begin{align*}
\mathbf{J}_{\mathcal{L}}^{(\varphi_1,\varphi_2,\eta)}
=
\begin{pmatrix}
\cos\varphi_1\!\! & \!\!-\sin\varphi_1 \\
\sin\varphi_1\!\! & \!\!\cos\varphi_1
\end{pmatrix}
\begin{pmatrix}
L\!\!\! & \!\!0 \\
0\!\!\! & \!\!\eta L
\end{pmatrix}
\begin{pmatrix}
\cos\varphi_2\!\! & \!\!-\sin\varphi_2 \\
\sin\varphi_2\!\! & \!\!\cos\varphi_2
\end{pmatrix}\!,
\end{align*}
where $\varphi_1,\varphi_2\in[0,2\pi)$ and $\eta\in[-1,1]$ are the additional three parameters.
Substituting these representations into Eq.~\eqref{eq:arbitrary_N} and collecting the seven parameters as $\boldsymbol{\zeta}=(u_1,u_2,v_1,v_2,\varphi_1,\varphi_2,\eta)$, the bound is obtained by evaluating%
\footnote{
We formulate the optimization problem using $\sup$ rather than $\max$ since the optimal value may not be attained within a bounded set.
For example, the softplus function $f(x)=\log(1+e^x)$ is 1-Lipschitz continuous, while its derivative, i.e., the sigmoid function, approaches $1$ only in the limit $x\to+\infty$.
This technical issue does not arise for the LipsAMs considered in this paper.
}
\begin{align}
\!\!B_{(L,\psi)}=\sup_{\boldsymbol{\zeta}\in\mathbb{R}^2\times\mathbb{R}_+^2\times[0,2\pi)^2\times[-1,1]}\|\mathbf{D}_1^{(\boldsymbol{\zeta})}\mathbf{J}_{\mathcal{L}}^{(\boldsymbol{\zeta})}\!+\mathbf{D}_2^{(\boldsymbol{\zeta})}\|_{\mathrm{op}},
\label{eq:worst_case_lipschitz_bound_2d}
\end{align}
where, for notational simplicity, $\mathbf{D}_1^{(\boldsymbol{\zeta})}$, $\mathbf{D}_2^{(\boldsymbol{\zeta})}$, and $\mathbf{J}_{\mathcal{L}}^{(\boldsymbol{\zeta})}$ denote the above parametrizations using the corresponding elements of $\boldsymbol{\zeta}$.
This quantity can be numerically estimated using an optimizer, such as a particle swarm solver \cite{kennedyParticleSwarmOptimization}.
Since Eq.~\eqref{eq:worst_case_lipschitz_bound_2d} is generally nonconvex and a solver may yield a suboptimal solution, we recommend running the solver multiple times and taking the maximum value.
For additional confidence, this value can be examined as in Sect.~\ref{sec:verify}.

\subsection{Lipschitz Bounds of the Proposed Architectures}
\label{sec:lips_bound_se_re}

We now evaluate the Lipschitz bounds of the proposed LipsAM architectures.
The results are summarized in Table~\ref{fig:architectures}.
For the estimator-based architectures (i.e., LipsAM-E and ReM-LipsAM-E), we derived closed-form formulas as follows.
\begin{prop}
For any $N\geq 2$, 
\label{thm:lips_bound_se_re}
\begin{align}
&\sup_{\mathcal{E}:\mathbb{R}_+^N\to\mathbb{R}^N,\,\operatorname{Lip}(\mathcal{E})\leq L}
\operatorname{Lip}(\mathcal{D}_{\mathcal{E}}^\mathrm{(Lips)})
=
\sqrt{L^2+1},
\\
&\sup_{\mathcal{E}:\mathbb{R}_+^N\to\mathbb{R}^N,\,\operatorname{Lip}(\mathcal{E})\leq L}
\operatorname{Lip}(\mathcal{D}_{\mathcal{E}}^\mathrm{(ReM\text{-}Lips)})
=
L+1.
\end{align}
\end{prop}
\begin{proof}
See Appendix~\ref{app:proof_thm_lips_bound_se_re}.
\end{proof}
\noindent 
For LipsAM-M, the nonlinear function $g$ (e.g., the sigmoid function) makes it difficult to obtain a closed-form formula.
We therefore apply the proposed two-step framework by setting $\psi$ as in Eq.~\eqref{eq:phiM} and computing Eq.~\eqref{eq:worst_case_lipschitz_bound_2d} numerically.

\subsection{Empirical Confirmation of the Lipschitz Bounds}
\label{sec:verify}

We numerically confirmed Theorem~\ref{thm:worst_case_lipschitz_bound} and the Lipschitz bounds obtained above.
We naively estimated the maximum Lipschitz constants of each LipsAM architecture by computing the following quantity using the Adam optimizer:
\begin{equation}
\sup_{\mathbf{z} \in \mathbb{C}^N, \boldsymbol{\theta} \in \Theta}\|\mathbf{J}_\mathcal{D}(\mathbf{z};\boldsymbol{\theta})\|_\mathrm{op},
\label{eq:naivebound}
\end{equation}
where $\mathcal{D}:\mathbb{C}^N\to\mathbb{C}^N$ represents LipsAMs (i.e., $\mathcal{D}\in\{\mathcal{D}_\mathcal{E}^\mathrm{(Lips)}, \mathcal{D}_\mathcal{E}^\mathrm{({ReM\text{-}Lips})}, \mathcal{D}_{(\mathcal{M},g,a,b)}^\mathrm{(Lips)}\}$), and $\boldsymbol{\theta} \in \Theta$ denotes the parameters of the learnable mapping $\mathcal{L}$ (i.e., $\mathcal{E}$ and $\mathcal{M}$).
To control the Lipschitz constant of the learnable mapping $\mathcal{L}$, we implemented it with scaling by a factor of $L\in\{1/4,1/2,1,2,4\}$ and an \textit{almost orthogonal layer}\cite{prachAlmostOrthogonalLayersEfficient2022} as $\mathcal{L}(\mathbf{x})=(\mathbf{T}(L\mathbf{x})+\mathbf{b})$, where $\|\mathbf{T}\|_\mathrm{op}\leq 1$ and $\mathbf{b}\in\mathbb{R}^N$, resulting in $\operatorname{Lip}(\mathcal{L})\leq L$.
The dimensions were set to $N\in\{1,2,3,4,5\}$.
We generated 100 random initial-value candidates and selected the top 10 that yielded the largest values of $\|\mathbf{J}_\mathcal{D}(\mathbf{z};\boldsymbol{\theta})\|_\mathrm{op}$.
The learning rate of the Adam optimizer was set to $0.01$, and the optimization was run for a maximum of $500$ iterations.
To utilize automatic differentiation, a power iteration method with $5$ iterations was used to calculate the operator norm during the Adam optimization process.
After the optimization concluded, the maximum singular value was computed using the \texttt{svd} function in MATLAB.

Figure~\ref{fig:verification} shows the results.
The results in Sect.~\ref{sec:lips_bound_se_re} (solid lines) tightly bounded their naive estimates via Eq.~\eqref{eq:naivebound} (markers). 
The naive estimates achieved the bound when $N=2$ (circles) in many cases, supporting Theorem~\ref{thm:worst_case_lipschitz_bound}.
For $N=1$ (squares), the naive estimates for LipsAM-E and ReM-LipsAM-E remained at $1$ and failed to attain their analytical bounds, which arises from the piecewise-linear structure of the architectures used in this experiment, for which the objective function in Eq.~\eqref{eq:naivebound} becomes locally constant, causing the gradients to vanish and preventing maximization by Adam.
Note also that the Lipschitz constant of a trained LipsAM can be smaller than the bounds displayed by the solid lines.

\begin{figure}
    \centering
    \centering
    \centering
    \includegraphics[scale=0.5]{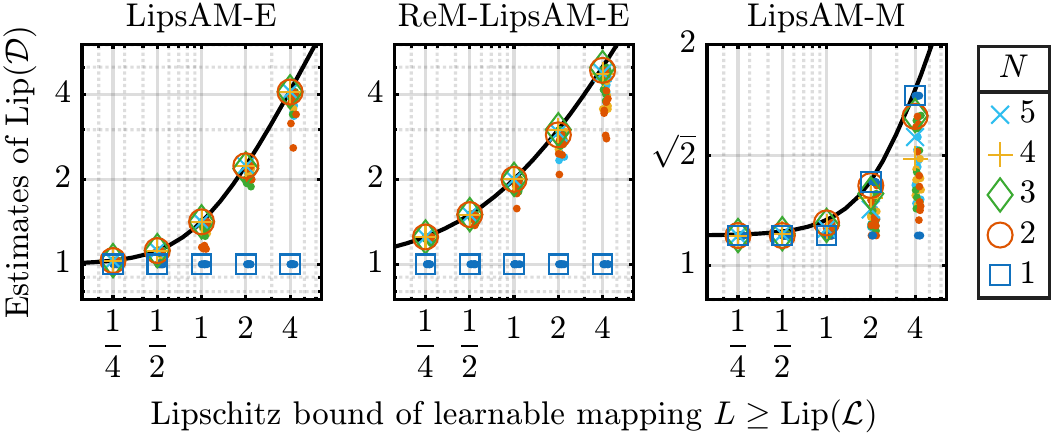}
    \caption{%
    Confirmation of the Lipschitz bounds of the proposed LipsAM architectures.
    The markers show the naive estimates of the maximum Lipschitz constant obtained by Eq.~\eqref{eq:naivebound} and the Adam optimizer, where $N$ corresponds to the dimension of the inputs, small jittered dots show the results for the 10 initial values, and large markers show the maximum.
    The solid lines show the upper bounds obtained in Sect.~\ref{sec:lips_bound_se_re}. 
    The bounds for LipsAM-E and ReM-LipsAM-E are analytical (see Proposition~\ref{thm:lips_bound_se_re}), whereas that for LipsAM-M ($g=\mathrm{sigmoid}$, $a=1$, and $b=0$) was numerically estimated using the proposed framework in Sect.~\ref{sec:bounds}; We used the \texttt{particleswarm} function \cite{kennedyParticleSwarmOptimization} in MATLAB's Global Optimization Toolbox to solve Eq.~\eqref{eq:worst_case_lipschitz_bound_2d}. 
    For each setting, we ran \texttt{particleswarm} 10 times and took the maximum.
    }
    \label{fig:verification}
\end{figure}

\section{Convergent PnP Audio Signal Recovery as an Application of LipsAM}
\label{sec:algo}

As an application of the proposed LipsAMs, we construct a provably convergent PnP algorithm for audio signal recovery.
While various convergence guarantees have been developed for PnP algorithms, we focus on the result in \cite{ryuPlugandPlayMethodsProvably2019}.
We selected this algorithm because, in our preliminary experiments on speech dereverberation, this algorithm achieved favorable signal recovery performance.
Note that the applicability of LipsAMs is not limited to this algorithm, and they can be used to satisfy other convergence conditions.

\subsection{CoReM-LipsAM and a Convergent PnP Algorithm}

Now we have a Lipschitz-continuous DNN for processing complex-valued signals $\mathcal{D}_{\!\mathcal{A}}:\mathbb{C}^N\to\mathbb{C}^N$ and an upper bound on its Lipschitz constant $B\geq\operatorname{Lip}(\mathcal{D}_{\!\mathcal{A}})$.
By combining them, we define \textbf{CoReM-LipsAM} as follows:
\begin{equation}
    \text{CoReM-LipsAM:}\quad\mathcal{D}_{(\mathcal{A},C)}^{\mathrm{(CoReM\text{-}Lips)}} = \mathrm{Id}-\underset{=\mathcal{R}}{\underbrace{(C/B)\mathcal{D}_{\!\mathcal{A}}^\mathrm{(Lips)}}},
    \label{eq:CRLipsAM}
\end{equation}
where $C\in\mathbb{R}_+$ is a parameter that controls the Lipschitz constant of $\mathcal{R}$.

Using CoReM-LipsAM, we propose a PnP algorithm for audio signal processing based on the PnP-ADMM in Eq.~\eqref{eq:plug-and-play} and the convergence result established in \cite{ryuPlugandPlayMethodsProvably2019}.
We introduce an STFT operator $\mathbf{G}\in\mathbb{C}^{N\times T}$ to transform the time-domain signals into the time-frequency domain.
By plugging the time domain denoiser $\mathcal{D}:\mathbb{R}^T\to\mathbb{R}^T$, given by
\begin{equation}
\mathcal{D}(\mathbf{x})=\mathbf{G}^\mathsf{H}\mathcal{D}_{(\mathcal{A},C)}^{\text{(CoReM-Lips)}}(\mathbf{G}\mathbf{x}),
\end{equation}
into Eq.~\eqref{eq:plug-and-play}, we propose the following algorithm.
\begin{equation}
\label{eq:plug-and-play-LipsAM}
\left\lfloor\quad
\begin{aligned}
\hat{\mathbf{x}}^{[k+1]} &= \mathbf{G}^\mathsf{H}\mathcal{D}_{(\mathcal{A},C)}^{\mathrm{(CoReM\text{-}Lips)}}(\mathbf{G}(\boldsymbol{\xi}^{[k]}-\boldsymbol{\upsilon}^{[k]})),\\
\boldsymbol{\xi}^{[k+1]} &= \mathrm{prox}_{\alpha \mathfrak{F}_\mathbf{y}}(\hat{\mathbf{x}}^{[k+1]}+\boldsymbol{\upsilon}^{[k]}),\\
\boldsymbol{\upsilon}^{[k+1]} &= \boldsymbol{\upsilon}^{[k]}+\hat{\mathbf{x}}^{[k+1]}-\boldsymbol{\xi}^{[k+1]},
\end{aligned}
\right.
\end{equation}

The next theorem guarantees the convergence of this algorithm under the assumption that a Parseval-tight window\cite{janssenCharacterizationComputationCanonical2002}
is used for the STFT operator, i.e., $\mathbf{G}^\mathsf{H}\mathbf{G}=\mathbf{I}$ holds.

\begin{thm}[Convergent PnP using a CoReM-LipsAM]
Let $\mathcal{D}_{(\mathcal{A},C)}^{\mathrm{(CoReM\text{-}Lips)}} = \mathrm{Id}-(C/B)\mathcal{D}_{\!\mathcal{A}}^\mathrm{(Lips)}$, where
$C\in[0,1)$ and $\mathcal{D}_{\!\mathcal{A}}^\mathrm{(Lips)}:\mathbb{C}^N\to\mathbb{C}^N$ is $B$-Lipschitz continuous. 
Let $\mathfrak{F}_\mathbf{y}: \mathbb{R}^T \to \mathbb{R}$ be $\mu$-strongly convex, 
and $\mathbf{G}\in\mathbb{C}^{N\times T}$ and $\alpha>0$ satisfy the following conditions:
\begin{equation}
    \mathbf{G}^\mathsf{H}\mathbf{G}=\mathbf{I},\qquad\frac{C}{\mu(1+C-2C^2)} < \alpha.
\end{equation}
Then the sequence $(\hat{\mathbf{x}}^{[k]})_{k\in\mathbb{N}}$ generated by Eq.~\eqref{eq:plug-and-play-LipsAM} converges to its unique fixed point.
\label{thm:contractionTF}
\end{thm}
\begin{proof}
Using the Parseval-tightness $\mathbf{G}^\mathsf{H}\mathbf{G}=\mathbf{I}$, we have
\begin{equation*}
\begin{aligned}\mathbf{G}^\mathsf{H}\circ\mathcal{D}_{(\mathcal{A},C)}^\mathrm{(CoReM\text{-}Lips)} \circ \mathbf{G}&=\mathrm{Id}-\underset{\mathcal{R}:\mathbb{R}^T\to\mathbb{R}^T}{\underbrace{C\mathbf{G}^\mathsf{H} \circ(\mathcal{D}_{\!\mathcal{A}}^\mathrm{(Lips)}/B)\circ\mathbf{G}}}.
    \end{aligned}
\end{equation*}
Then $\operatorname{Lip}(\mathcal{R})\leq C<1$.
Therefore, the operator $\mathcal{D}=\mathbf{G}^\mathsf{H}\circ\mathcal{D}_{(\mathcal{A},C)}^\mathrm{(CoReM\text{-}Lips)} \circ \mathbf{G}$ satisfies 
Eq.~(A) in \cite{ryuPlugandPlayMethodsProvably2019}.
Together with the other assumptions, \cite[Corollary~3]{ryuPlugandPlayMethodsProvably2019} guarantees the convergence to a unique fixed point.
\end{proof}

\subsection{Estimator- and Masking-based CoReM-LipsAMs}

We next propose CoReM-LipsAM architectures based on LipsAM-E and LipsAM-M introduced in Sect.~\ref{sec:LipsAM}.
We first propose \textbf{CoReM-LipsAM-E} using LipsAM-E in Eq.~\eqref{eq:signallim}.
As shown in Proposition~\ref{thm:lips_bound_se_re}, the Lipschitz constant of LipsAM-E is bounded by $B_\mathcal{E}(L)=\sqrt{L^2+1}>0$, where $L\geq\operatorname{Lip}(\mathcal{E})$.
Substituting them into Eq.~\eqref{eq:CRLipsAM} yields
\begin{align}
    &\text{CoReM-LipsAM-E}:\quad\mathcal{D}_{(\mathcal{E},C)}^{\text{(CoReM-Lips)}}(\mathbf{z}) = \nonumber\\
    &\qquad\mathbf{z}-(C/B_\mathcal{E}(L))\cdot\bigl(\min(\mathcal{E}(|\mathbf{z}|),|\mathbf{z}|)\bigr)_+\odot\operatorname{sign}(\mathbf{z}).
    \label{eq:CoReM-LipsAM-E}
\end{align}
Similarly, we propose \textbf{CoReM-LipsAM-M}, which is defined with a $L$-Lipschitz-continuous mapping $\mathcal{M}$ as%
\footnote{
The reason for the sign inversion before $g$ is the same as in Sect.~\ref{sec:residualmap}.
}
\begin{align}
&\text{CoReM-LipsAM-M}:
\mathcal{D}_{(\mathcal{M},g,a,b,C)}^{\mathrm{(\text{CoReM-Lips})}}(\mathbf{z})
= \\
&\qquad \mathbf{z}-\frac{C}{B_\mathcal{M}(L,a,b)}\cdot\underset{{}=\widetilde{\mathcal{D}}^\mathrm{(Lips)}_{(\mathcal{M},g,a,b)}(\mathbf{z})}{\underbrace{ G\Bigl(-\max(\mathcal{M}(|\mathbf{z}|),\, a|\mathbf{z}| + b \mathbf{1})\Bigr)\odot \mathbf{z}}},  \nonumber
\label{eq:CoReM-LipsAM-M}
\end{align}
where $B_\mathcal{M}(L,a,b)>0$ denotes an upper bound on
$\operatorname{Lip}(\widetilde{\mathcal{D}}^\mathrm{(Lips)}_{(\mathcal{M},g,a,b)})$
and can be numerically estimated by setting $\psi(u,v) = g(-\max(u,a v + b))\cdot v$ and computing Eq.~\eqref{eq:worst_case_lipschitz_bound_2d}.

Note that both CoReM-LipsAM-E and CoReM-LipsAM-M can represent the soft-thresholding operator, which is the proximity operator of the $\ell_1$-norm, given by  
\begin{equation}
    \text{\rm{Soft-Thresholding}:}\quad
    \mathcal{S}_\tau(\mathbf{z})
    =
    (|\mathbf{z}|-\tau)_+\odot\operatorname{sign}(\mathbf{z}),
    \label{eq:soft_thresholding}
\end{equation}
where $\tau>0$.
By setting $\mathcal{E}(\mathbf{x})=\tau\mathbf{1}$ (whose Lipschitz constant is $0$) and $C=1$, CoReM-LipsAM-E reduces to the soft-thresholding operator.
CoReM-LipsAM-M also reduces to the soft-thresholding operator with $g(t) = 1/\max\left(-t,1\right)$, $\mathcal{M}(\mathbf{x})=\mathbf{0}$, $a>0$, $b=0$ and $C=1$.
In this case, $a>0$ determines the threshold as $\tau = 1/a$.

\subsection{Practical Considerations for Training CoReM-LipsAMs}
\label{sec:training_corem_lipsams}

The Lipschitz bound $L\geq\operatorname{Lip}(\mathcal{L})$ should be explicitly handled during training, where $\mathcal{L}\in\{\mathcal{E},\mathcal{M}\}$ is the learnable mapping.
To this end, we construct the learnable mapping $\mathcal{L}$ using a learnable $1$-Lipschitz-continuous mapping
$\widetilde{\mathcal{L}}:\mathbb{R}_+^N\to\mathbb{R}^N$ and pre-scaling by $L=\operatorname{sigmoid}(s_L)\in(0,1)$ as%
\footnote{
The scaling factor $L$ is applied before $\widetilde{\mathcal{L}}$ rather than after it because, under the post-scaling form, $\mathcal{L}(\mathbf{x})=L\widetilde{\mathcal{L}}(\mathbf{x})$,
$L=0$ would necessarily yield $\mathcal{L}(\mathbf{x})=\mathbf{0}$ for all $\mathbf{x}\in\mathbb{R}_+^N$, thereby excluding nonzero constant mappings whose Lipschitz constants are zero.
The learned values of $L$ for the best-performing models in our experiment (in Sect.~\ref{sec:expt}, marked in Fig.~\ref{fig:sisnr_1}) were $0.48$ for CoReM-LipsAM-E and $0.14$ for CoReM-LipsAM-M.
}
\begin{equation}
\mathcal{L}(\mathbf{x})
=
\widetilde{\mathcal{L}}(L\,\mathbf{x}),
\end{equation}
where $s_L\in\mathbb{R}$ is a learnable parameter.
The mapping $\widetilde{\mathcal{L}}$ can be implemented using existing methods for constructing $1$-Lipschitz-continuous DNNs, such as structurally $1$-Lipschitz-continuous layers \cite{prach1LipschitzLayersCompared2024,meunierDynamicalSystemPerspective2022a,prachAlmostOrthogonalLayersEfficient2022,liPreventingGradientAttenuation2019}.
This construction guarantees $\operatorname{Lip}(\mathcal{L})\leq L=\operatorname{sigmoid}(s_L)$ and allows $L$ to be jointly optimized with the parameters of $\widetilde{\mathcal{L}}$.

For CoReM-LipsAM-M, the bound $B_{\mathcal{M}}(L,a,b)$ must be numerically estimated for each $(L,a,b)$.
However, this procedure is incompatible with training using automatic differentiation.
Therefore, we propose to use a differentiable approximation for this bound during training. 
Let $\mathcal{G}=\mathcal{S}_L\times\mathcal{S}_a\times\mathcal{S}_b$ be a grid of parameters, where $\mathcal{S}_L\subset[0, 1]$, $\mathcal{S}_a\subset\mathbb{R}_{++}$, and $\mathcal{S}_b\subset\mathbb{R}$.
We precompute estimates of $B_{\mathcal{M}}(L,a,b)$ at every point in $\mathcal{G}$ using the proposed framework in Sect.~\ref{sec:bounds} and construct a differentiable approximation $\widetilde{B}(L,a,b)\approx B_{\mathcal{M}}(L,a,b)$ by using, e.g., cubic splines.
Moreover, to constrain the parameters $(a,b)$ to be in the ranges covered by $\mathcal{G}$, we reparameterize them as 
\begin{align}
a &= \min\mathcal{S}_a
+\left(\max\mathcal{S}_a-\min\mathcal{S}_a\right)
\mathrm{sigmoid}(s_a),\label{eq:reparam_a}\\
b &= \min\mathcal{S}_b
+\left(\max\mathcal{S}_b-\min\mathcal{S}_b\right)
\mathrm{sigmoid}(s_b),\label{eq:reparam_b}
\end{align}
where $s_a,s_b\in\mathbb{R}$ are learnable parameters.
After training, the approximated bound is replaced by the numerical estimate obtained by the framework in Sect.~\ref{sec:bounds} for the learned $(L,a,b)$.

\section{Experiment on Speech Dereverberation}
\label{sec:expt}

\subsection{Problem Setting and Algorithm}
This experiment considers the recovery of a speech signal $\mathbf{x}\in\mathbb{R}^T$ from its reverberant and noisy observation $\mathbf{y}=\mathbf{H}\mathbf{x} + \mathbf{n}\in\mathbb{R}^T$, where $T$ is the signal length in the time domain,
$\mathbf{H}\in\mathbb{R}^{T \times T}$ is a convolution matrix for a room impulse response $\mathbf{h}\in\mathbb{R}^T$, 
and $\mathbf{n}\in\mathbb{R}^T$ is additive Gaussian noise. 
The level of $\mathbf{n}$ was set to $-40$\,dB with respect to the reverberant speech. 
For the STFT operator $\mathbf{G}\in\mathbb{C}^{N\times T}$, a Parseval-tight window computed from a Hann window \cite{janssenCharacterizationComputationCanonical2002} was used with a window length of 256 and a hop size of 128. 
Ten pairs of a source signal and an impulse response $(\mathbf{x},\mathbf{h})$ were randomly chosen from the \texttt{test-clean} in LibriTTS-R \cite{koizumiLibriTTSRRestoredMultiSpeaker2023} and the BUT reverb database \cite{szokeBuildingEvaluationReal2019}, respectively.
The signals were resampled to 8\,kHz and then trimmed to 64 frames in the STFT domain.

The algorithm in Eq.~\eqref{eq:plug-and-play-LipsAM} is applicable for this task by setting the data fidelity function as follows:
\begin{equation}
    \mathfrak{F}_\mathbf{y}(\hat{\mathbf{x}}) = \frac{1}{2}\left\|\mathbf{H}\hat{\mathbf{x}} - \mathbf{y}\right\|_2^2 + \frac{\gamma}{2}\|\hat{\mathbf{x}}\|_2^2,
\end{equation}
where the ridge term $(\gamma/2)\|\hat{\mathbf{x}}\|_2^2$ is added to make the data-fidelity function at least $\gamma$-strongly convex (the strong convexity is required in Theorem~\ref{thm:contractionTF}).
The proximity operator with parameter $\alpha\geq0$ for this term is given by
$\mathrm{prox}_{\alpha \mathfrak{F}_\mathbf{y}}(\mathbf{v})=((1+\alpha\gamma)\mathbf{I}+\alpha\mathbf{H}^\mathsf{T}\mathbf{H})^{-1}(\mathbf{v}+\alpha\mathbf{H}^\mathsf{T}\mathbf{y})$.
The parameters of the algorithm were set to $\gamma=0.005$, $C=0.99$, and $\alpha=1.01 C/(\gamma(1+C-2C^2))$, so that Theorem~\ref{thm:contractionTF} applies.
The performance was evaluated by the scale-invariant SNR (SI-SNR) \cite{rouxSDRHalfbakedWell2019}.
We also computed $\Delta\mathbf{x}^{[k]}=\|\hat{\mathbf{x}}^{[k+1]}-\hat{\mathbf{x}}^{[k]}\|_2$, which converges to 0 when the algorithm converges to a point.

\subsection{Denoisers}

We trained the proposed CoReM-LipsAM-E and CoReM-LipsAM-M on speech denoising tasks using additive Gaussian noise (which is a standard procedure for PnP methods).

For comparison, we also trained two conventional residual-map-based AMs (ReM-AMs), defined as
\begin{equation}
    \text{ReM-AM:}\quad
    \mathcal{D}_{\!\mathcal{A}}^{(\mathrm{ReM})}
    =
    \mathrm{Id}-\mathcal{D}_{\!\mathcal{A}},
    \label{eq:ReM-AM}
\end{equation}
where AM-E in Eq.~\eqref{eq:signal} and AM-M in Eq.~\eqref{eq:masking} were used for $\mathcal{D}_{\!\mathcal{A}}$.
Since ReM-AMs are generally not Lipschitz continuous, convergence of the PnP algorithm with them is not guaranteed.

\begin{figure}
    \centering
    \includegraphics[scale=0.5]{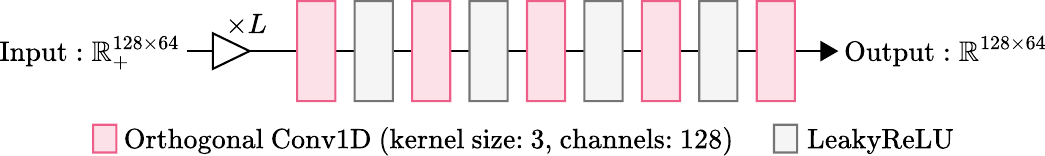}
    \caption{%
    The architecture of learnable mapping $\mathcal{E}$ and $\mathcal{M}$ for CoReM-LipsAMs. 
    The value $L\in(0,1)$ is parametrized as $L=\operatorname{sigmoid}(s_L)$, where $s_L\in\mathbb{R}$ is a learnable parameter. 
    This structure ensures $\operatorname{Lip}(\mathcal{E}), \operatorname{Lip}(\mathcal{M})\leq L$.
    }
    \label{fig:network}
\end{figure}

The architecture of the learnable mapping $\mathcal{L}\in\{\mathcal{E},\mathcal{M}\}$ for CoReM-LipsAM is illustrated in Fig.~\ref{fig:network}.
Following the construction in Sect.~\ref{sec:training_corem_lipsams}, it first applies pre-scaling by $L=\operatorname{sigmoid}(s_L)$ ($s_L\in\mathbb{R}$), followed by five 1D orthogonal convolutional layers \cite{prachAlmostOrthogonalLayersEfficient2022} combined with Leaky ReLU (slope of $0.1$).
For ReM-AMs, we used standard convolutional layers and omitted pre-scaling.
The kernel size and the number of channels of the convolutional layers were set to $3$ and $128$, respectively.
For masking-based architectures, we set $g=\operatorname{sigmoid}$.
For training CoReM-LipsAM-M, a differentiable surrogate for $B_{\mathcal{M}}$ was constructed as in Sect.~\ref{sec:training_corem_lipsams}, where the parameter sets were chosen as
$\mathcal{S}_L=\{0,0.1,\ldots,1\}$,  
$\mathcal{S}_a=\{0.1,0.2,\ldots,1\}$,
and 
$\mathcal{S}_b=\{-5,-4,\ldots,5\}$.

We used \texttt{train-clean-100} in LibriTTS-R\cite{koizumiLibriTTSRRestoredMultiSpeaker2023} as the dataset for clean speech signals.
The SNR during training was set from $10$ to $30$\,dB in $2$\,dB steps.
Time-domain mean squared error (MSE) loss was minimized using the Adam optimizer with a learning rate of $0.001$.
The batch size was set to 32, and the models were trained for 10 epochs.
As a baseline, the soft-thresholding operator in Eq.~\eqref{eq:soft_thresholding} was also trained for each SNR setting in the same way as the DNNs, where the threshold $\tau\in\mathbb{R}_+$ was the only learnable parameter.

\subsection{Results}

The experimental results are shown in Fig.~\ref{fig:sisnr_1}.
The left panel shows the SI-SNR (higher is better).
Here, the horizontal axis corresponds to the SNR during training, where a lower SNR means a higher noise level during training.
The conventional ReM-AMs (dotted lines) achieved good performance in this experiment.
In particular, ReM-AM-M, which is based on time-frequency masking, yielded the highest SI-SNR.
The proposed CoReM-LipsAMs (solid lines) showed degraded performance compared to them.
This degradation may result from the architectural constraints imposed to guarantee Lipschitz continuity or contractivity of residual mappings.
Among the proposed architectures, CoReM-LipsAM-E achieved higher SI-SNR than CoReM-LipsAM-M.
Both CoReM-LipsAMs achieved higher SI-SNR than soft-thresholding (gray line).

For the best cases for each architecture (markers in the left panel), the transition of $\Delta\mathbf{x}^{[k]}$ is shown in the right panel in Fig.~\ref{fig:sisnr_1}.
The value $\Delta\mathbf{x}^{[k]}$ for ReM-AMs (dotted lines) did not decrease, suggesting that the algorithm did not converge.
In contrast, the proposed CoReM-LipsAMs (solid lines) reduced $\Delta\mathbf{x}^{[k]}$ to around $10^{-12}$, which is comparable to that obtained by soft-thresholding (gray line).
This supports the convergence result in Theorem~\ref{thm:contractionTF}.
Moreover, the CoReM-LipsAMs converged in fewer iterations than soft-thresholding, indicating their fast empirical convergence in this experiment.

\begin{figure}
    \includegraphics[scale=0.5]{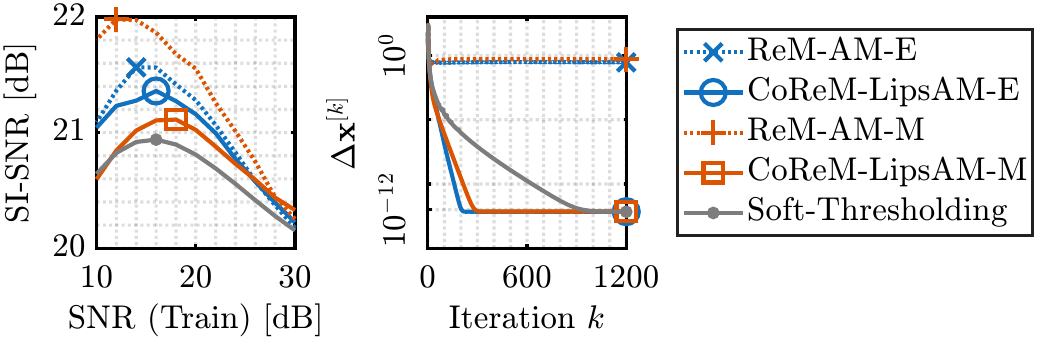}
    \caption{Performance of the PnP algorithm on the speech dereverberation task.
    Dotted lines indicate conventional DNN architectures that have no theoretical guarantee.
    Solid colored lines indicate the proposed CoReM-LipsAMs.
    The gray line is the result of soft-thresholding as a baseline.
    (Left) SI-SNR after 1200 iterations with respect to the SNR during training with denoising tasks.
    The best results for each architecture are marked.
    (Right) Transition of $\Delta\mathbf{x}^{[k]}$ for the marked results in the left panel.
    }
    \label{fig:sisnr_1}
\end{figure}

\section{Conclusion}
\label{sec:conc}

In this paper, we investigated LipsAMs, i.e., Lipschitz-continuous amplitude modifiers.
We established necessary and sufficient conditions for the Lipschitz continuity of AMs and developed DNN architectures of LipsAM.
We also developed an efficient framework for evaluating Lipschitz constants by reducing the input dimensionality to two.
These results provide a theoretical foundation for constructing and analyzing DNN-based audio signal processing methods with a theoretical guarantee.
Building on these results, we proposed CoReM-LipsAMs and a convergent PnP algorithm for audio signal recovery.
An experiment on speech dereverberation confirmed the convergence behavior of the proposed methods.

Our experiment showed that the proposed methods underperformed their conventional counterparts in the dereverberation task and provided only limited improvements over soft-thresholding.
This result indicates a trade-off between theoretical guarantees and processing performance, as enforcing Lipschitz continuity and/or contractivity on the residual maps restricts the expressive power of DNNs.
While we adopted the contractivity of the residual map as a condition on DNNs, various other convergence conditions and algorithms have been developed for PnP methods. Exploring such alternatives may lead to improved signal recovery performance while guaranteeing convergence.
In addition, we evaluated the proposed methods only on speech dereverberation, although their applicability is not limited to this task. Evaluating them on a broader range of signal recovery tasks is therefore an important direction for future work.
More broadly, since Lipschitz continuity is useful beyond PnP algorithms, LipsAMs may also be applicable to other signal processing methods that require Lipschitz continuity or Lipschitz constants of DNNs.

\appendices

\section{Proof of Theorem \ref{thm:lipsam}}
\label{app:proof_thm_lipsam}

We show that an AM $\mathcal{D}_{\!\mathcal{A}}=\mathcal{A}(|\mathbf{z}|)\odot\operatorname{sign}(\mathbf{z})$ is Lipschitz continuous $\Leftrightarrow$ $\mathcal{D}_{\!\mathcal{A}}$ is a LipsAM defined in Definition~\ref{def:lipsam}.

\begin{proof}
($\Leftarrow$) 
Take arbitrary $\mathbf{z}$, $\mathbf{w} \in \mathbb{C}^N$ and represent them in polar form as $\mathbf{z} = \mathbf{x}\odot \exp(\rm{i}\boldsymbol{\varphi})$, $\mathbf{w} = \mathbf{y}\odot \exp(\rm{i}\boldsymbol{\phi})$, respectively,
where $\mathbf{x}$, $\mathbf{y} \in \mathbb{R}_+^N$, $\boldsymbol{\varphi}$, $\boldsymbol{\phi} \in [0,2\pi)^N$, and $\exp(\cdot)$ is the element-wise exponential function.
Then we have $|z_n - w_n|^2 = |x_n - y_n|^2 + 2 x_n y_n (1 - \cos (\varphi_n - \phi_n))$ for all $n\in[N]$.
This with Eqs.~\eqref{eq:cond1} and \eqref{eq:cond2} yields
$\|\mathcal{D}_{\!\mathcal{A}}(\mathbf{z})-\mathcal{D}_{\!\mathcal{A}}(\mathbf{w})\|_2^2
=\|\mathcal{A}(\mathbf{x})-\mathcal{A}(\mathbf{y})\|_2^2+2\sum_{n=1}^N(\mathcal{A}(\mathbf{x}))_n(\mathcal{A}(\mathbf{y}))_n(1-\cos(\varphi_n-\phi_n))
\leq L_1^2\|\mathbf{x}-\mathbf{y}\|_2^2+2L_2^2\sum_{n=1}^N x_ny_n(1-\cos(\varphi_n-\phi_n))
\leq\max(L_1^2,L_2^2)(\|\mathbf{x}-\mathbf{y}\|_2^2+2\sum_{n=1}^N x_ny_n(1-\cos(\varphi_n-\phi_n)))
\leq(\max(L_1,L_2)\|\mathbf{z}-\mathbf{w}\|_2)^2
$.
Therefore, $\mathcal{D}_{\!\mathcal{A}}$ is $\max(L_1,L_2)$-Lipschitz continuous.

($\Rightarrow$) 
The necessity of Lipschitz continuity of $\mathcal{A}$, i.e., Eq.~\eqref{eq:cond1}, is trivial.
The necessity of the condition in Eq.~\eqref{eq:cond2} can be proven by showing its contraposition.
Without loss of generality, we divide the input and output of $\mathcal{A}$ into their first components and 
the remaining $N-1$ components as 
$\mathcal{A}(\mathbf{x})=(\mathcal{A}_1(x_1,\mathbf{x}_{2:}),\mathcal{A}_2(x_1,\mathbf{x}_{2:}))$, where $\mathcal{A}_1:\mathbb{R}_+\times\mathbb{R}_+^{N-1}\to\mathbb{R}_+$ and $\mathcal{A}_2:\mathbb{R}_+\times\mathbb{R}_+^{N-1}\to\mathbb{R}_+^{N-1}$.
Let $\mathcal{D}_{\!\mathcal{A}}:\mathbb{C}^N\to\mathbb{C}^N$ be an AM given by $\mathcal{D}_{\!\mathcal{A}}(\mathbf{z})=\bigl(\mathcal{A}_1(|z_1|,|\mathbf{z}_{2:}|),\mathcal{A}_2(|z_1|,|\mathbf{z}_{2:}|)\bigr)\odot\operatorname{sign}(\mathbf{z})$.
Assume that $\mathcal{D}_{\!\mathcal{A}}$ is Lipschitz continuous,
and, for any $L_2\geq 0$, there exists $\mathbf{x}=(x_1,\mathbf{x}_{2:})\in\mathbb{R}_+\times \mathbb{R}_+^{N-1}$ such that $(\mathcal{A}_1(x_1,\mathbf{x}_{2:})) > L_2 x_1$.
Taking $\mathbf{z}= (x_1,\mathbf{x}_{2:})$ and $\widetilde{\mathbf{z}} = (-x_1,\mathbf{x}_{2:})$ in $\mathbb{R}\times\mathbb{R}_+^{N-1}\subset\mathbb{C}\times\mathbb{C}^{N-1}$, we have
\begin{equation*}
\begin{aligned}
&\|\mathcal{D}_{\!\mathcal{A}}(\mathbf{z})-
\mathcal{D}_{\!\mathcal{A}}(\widetilde{\mathbf{z}})\|_2^2= 
|\mathcal{A}_1(x_1,\mathbf{x}_{2:}) - (-\mathcal{A}_1(x_1,\mathbf{x}_{2:}))|^2
\\
{}&\qquad\qquad\qquad\qquad\qquad+\|\mathcal{A}_2(x_1,\mathbf{x}_{2:}) - \mathcal{A}_2(x_1,\mathbf{x}_{2:})\|_2^2
\\
&{}= (2\mathcal{A}_1(x_1,\mathbf{x}_{2:}))^2> L_2^2(2x_1)^2
= L_2^2\|{\mathbf{z}-\widetilde{\mathbf{z}}\|_2^2}.
\end{aligned}
\end{equation*}
This contradicts Lipschitz continuity of $\mathcal{D}_{\!\mathcal{A}}$, indicating that if $\mathcal{D}_{\!\mathcal{A}}$ is Lipschitz continuous, then $(\mathcal{A}_1(x_1,\mathbf{x}_{2:}))\leq L_2\,x_1$ must hold for all $\mathbf{x}\in\mathbb{R}_+^N$ for some constant $L_2\geq 0$, which is exactly the statements in Eq.~\eqref{eq:cond2} in Definition~\ref{def:lipsam}.
\end{proof}

\section{Proof of Theorem \ref{prop:zero_preserving}}
\label{app:proof_prop_zero_preserving}

\begin{proof}
It suffices to show the equivalence between Eq.~\eqref{eq:cond2} and Eq.~\eqref{eq:zero_preserving} under the assumption that $\mathcal{A}$ is Lipschitz continuous.
Assume that for some $L_2\ge 0$, Eq.~\eqref{eq:cond2} holds for any $\mathbf{x}\in\mathbb{R}_+^N$ and $n\in[N]$.
Then, if $x_n=0$, we have $0\le (\mathcal{A}(\mathbf{x}))_n \le L_2 x_n = 0$, which implies Eq.~\eqref{eq:zero_preserving}.

Conversely, assume that Eq.~\eqref{eq:zero_preserving} holds and $\mathcal{A}$ is $L_1$-Lipschitz continuous for some $L_1\ge 0$.
Fix arbitrary $\mathbf{x}\in\mathbb{R}_+^N$ and $n\in[N]$.
Let $\tilde{\mathbf{x}}\in\mathbb{R}_+^N$ be defined by
$\tilde{x}_n=0$ and
$\tilde{x}_k=x_k$ for all $k\neq n$.
Using $(\mathcal{A}(\tilde{\mathbf{x}}))_n=0$, we have
$(\mathcal{A}(\mathbf{x}))_n= (\mathcal{A}(\mathbf{x}))_n-(\mathcal{A}(\tilde{\mathbf{x}}))_n\le \|\mathcal{A}(\mathbf{x})-\mathcal{A}(\tilde{\mathbf{x}})\|_2\le L_1\|\mathbf{x}-\tilde{\mathbf{x}}\|_2
= L_1 x_n$.
Hence Eq.~\eqref{eq:cond2} holds with $L_2=L_1$. 
\end{proof}

\section{Proof of Theorem \ref{thm:mask_sufficient}}
\label{sec:proof_thm:mask_sufficient}

\begin{proof}
Let $\mathcal{A}(\mathbf{x}) = G(\mathcal{M}(\mathbf{x})) \odot \mathbf{x}$.
The zero-preserving property of $\mathcal{A}$ is satisfied, i.e., for any $\mathbf{x}\in\mathbb{R}_+^N$ with $x_n=0$, 
it holds that $(\mathcal{A}(\mathbf{x}))_n=0$.
By Theorem~\ref{prop:zero_preserving}, it remains to prove the Lipschitz continuity of $\mathcal{A}$ to show that $\mathcal{D}_{\!\mathcal{A}}(\mathbf{z})=\mathcal{A}(|\mathbf{z}|)\odot\operatorname{sign}(\mathbf{z})$ is a LipsAM.
Take arbitrary $\mathbf{x}\in\mathbb{R}_+^N$ and let
$\mathbf{\Phi}(\mathbf{x})= 
\mathrm{diag}(x_1\cdot(\mathrm{d}g/\mathrm{d}t)((\mathcal{M}(\mathbf{x}))_1), \ldots, x_N\cdot(\mathrm{d}g/\mathrm{d}t)((\mathcal{M}(\mathbf{x}))_N))\in\mathbb{R}^{N\times N}$. 
Then we have
$\mathbf{J}_{\mathcal{A}}(\mathbf{x})=\mathrm{diag}\bigl(G(\mathcal{M}(\mathbf{x}))\bigr)+\mathbf{\Phi}(\mathbf{x})\,\mathbf{J}_{\mathcal{M}}(\mathbf{x})$. We now derive an upper bound on the operator norm of this Jacobian.
Fix arbitrary $\epsilon>0$.
Then Condition~(i) implies that there exist constants $C_\mathcal{M}\geq0$ (independent of $n$ and $\epsilon$) and $R_\epsilon\geq0$ (independent of $n$) such that
$|x_n|\ge R_\epsilon
\Rightarrow
\left|x_n/(\mathcal{M}(\mathbf{x}))_n\right|
\le C_{\mathcal{M}}+\epsilon$ for all $n\in[N]$.
By Condition~(ii), there exists a constant $C_g\geq0$ such that 
$|t\cdot (\mathrm{d}g/\mathrm{d}t)(t)|\leq C_g$ for every $t\in\mathbb{R}$.
Hence, if $|x_n|\ge R_\epsilon$, then 
$|x_n\cdot(\mathrm{d}g/\mathrm{d}t) ((\mathcal{M}(\mathbf{x}))_n)|=|x_n/(\mathcal{M}(\mathbf{x}))_n|\cdot|(\mathcal{M}(\mathbf{x}))_n\cdot (\mathrm{d}g/\mathrm{d}t) ((\mathcal{M}(\mathbf{x}))_n)|
\le (C_{\mathcal{M}}+\epsilon)C_g$ holds.
On the other hand, if $|x_n|<R_\epsilon$, then we have $|x_n\cdot (\mathrm{d}g/\mathrm{d}t)((\mathcal{M}(\mathbf{x}))_n)|
\le R_\epsilon\,\operatorname{Lip}(g)$.
Combining the above inequalities, 
we have $\|\mathbf{\Phi}(\mathbf{x})\|_{\mathrm{op}}\leq \max\left\{(C_{\mathcal{M}}+\epsilon)C_g, R_\epsilon\,\operatorname{Lip}(g)\right\}$ for any $\mathbf{x}\in\mathbb{R}_+^N$.
Therefore we obtain
\begin{align*}
\|\mathbf{J}_{\mathcal{A}}\|_{\mathrm{op}}
&\le \|\mathrm{diag}(G(\mathcal{M}(\mathbf{x})))\|_{\mathrm{op}}
   + \|\mathbf{\Phi}(\mathbf{x})\|_{\mathrm{op}}\,\|\mathbf{J}_{\mathcal{M}}(\mathbf{x})\|_{\mathrm{op}} \\
&\le B_g
 + \max\left\{(C_{\mathcal{M}}+\epsilon)C_g, R_\epsilon\,\operatorname{Lip}(g)\right\}\cdot \operatorname{Lip}(\mathcal{M}).\nonumber
\end{align*}
Hence the operator norm of $\mathbf{J}_{\mathcal{A}}$ is bounded for any $\mathbf{x}\in\mathbb{R}_+^N$,
and therefore $\mathcal{A}$ is Lipschitz continuous. 
\end{proof}

\section{Proof of Theorem \ref{thm:lipsam_sm}}
\label{app:proof_thm_lipsam_sm}

\begin{proof}
Let $\widetilde{\mathcal{M}}:\mathbb{R}_+^N\to\mathbb{R}^N:\mathbf{x}\mapsto\max(\mathcal{M}(\mathbf{x}),\, a \mathbf{x} + b \mathbf{1})$,
 and rewrite $\mathcal{D}_{(\mathcal{M},g,a,b)}^\mathrm{(Lips)}(\mathbf{z})
= G(\widetilde{\mathcal{M}}(|\mathbf{z}|)) \odot \mathbf{z}$.
The mapping $\widetilde{\mathcal{M}}$ is Lipschitz continuous since the element-wise maximum of two Lipschitz-continuous functions is also Lipschitz continuous.
Take arbitrary $n\in[N]$ and arbitrary sequence $(\mathbf{x}^{[k]})_{k\in\mathbb{N}}$ in $\mathbb{R}_+^N$ such that $x_n^{[k]}\to\infty$.
Then we have
$(\widetilde{\mathcal{M}}(\mathbf{x}^{[k]}))_n \geq a x_n^{[k]}+b$ and, for sufficiently large $k$, we have $a x_n^{[k]}+b>0$ $(a>0$, $b\in\mathbb{R})$.
Hence, for any $n\in[N]$, we have
\begin{equation}
\limsup_{k\to\infty}
\left|
\frac{x_n^{[k]}}{(\widetilde{\mathcal{M}}(\mathbf{x}^{[k]}))_n}
\right|
\le
\lim_{k\to\infty}\frac{x_n^{[k]}}{a x_n^{[k]}+b}
=
\frac{1}{a}.
\end{equation}
Therefore $\widetilde{\mathcal{M}}$ satisfies Condition~(i) with $C_\mathcal{M}=1/a$.
Since $g$ is Lipschitz continuous and satisfies Condition~(ii) by assumption,
Theorem~\ref{thm:mask_sufficient} guarantees that $\mathcal{D}_{(\mathcal{M},g,a,b)}^\mathrm{(Lips)}$ is a LipsAM.
\end{proof}

\section{Looseness of the Trivial Bound}
\label{sec:loose}
For a LipsAM $\mathcal{D}_{(\mathcal{L},\psi)}$ represented in the general form in Eq.~\eqref{eq:generalLipsAM}, we have an upper bound on its Lipschitz constants.

\begin{prop}
\label{prop:bound}
Let $\mathcal{D}_{(\mathcal{L},\psi)}$ be a LipsAM given in Eq.~\eqref{eq:generalLipsAM}.
Then $\operatorname{Lip}(\mathcal{D}_{(\mathcal{L},\psi)})\leq \operatorname{Lip}(\psi)\cdot\sqrt{(\operatorname{Lip}(\mathcal{L}))^2+1}$.
\end{prop}
\begin{proof}
Let $\widetilde{\mathcal{L}}:\mathbb{R}_+^N\to\mathbb{R}^{N}\times\mathbb{R}^N_+:\mathbf{x}\mapsto(\mathcal{L}(\mathbf{x}),\mathbf{x})$ and $\widetilde{\Psi}:\mathbb{R}^{N}\times\mathbb{R}^N_+\to\mathbb{R}_+^N$ apply $\psi:\mathbb{R}\times\mathbb{R}_+\to\mathbb{R}_+$ element-wise to the $N$ pairs of the input.
Then, $\mathcal{A}_{(\mathcal{L},\psi)}$ in Eq.~\eqref{eq:generalLipsAM} is $\widetilde{\Psi}\circ\widetilde{\mathcal{L}}$. 
Therefore, $\operatorname{Lip}(\widetilde{\Psi})=\operatorname{Lip}(\psi)$ and $\operatorname{Lip}(\mathrm{\widetilde{\mathcal{L}}})\leq\sqrt{(\operatorname{Lip}(\mathcal{L}))^2+1}$ yield the bound.
\end{proof}

However, this bound is not tight; ReM-LipsAM-E in Eq.~\eqref{eq:ReMLipsAME} is an example.
The function $\psi(u,v)=(v-(u)_+)_+$ in Eq.~\eqref{eq:phiREME} gives $\operatorname{Lip}(\psi)=\sqrt{2}$, and hence $\operatorname{Lip}(\mathcal{D}_{\mathcal{E}}^{(\mathrm{ReM\text{-}Lips)}})\leq \sqrt{2}\sqrt{(\operatorname{Lip}(\mathcal{E}))^2+1}$ by Proposition~\ref{prop:bound}.
At the same time, triangle inequality and the nonexpansiveness of ReLU gives, for any $\mathbf{x},\mathbf{y}\in\mathbb{R}_+^N$,  $\|\mathcal{A}_{\mathcal{E}}^\mathrm{(ReM\text{-}Lips)}(\mathbf{x})-\mathcal{A}_{\mathcal{E}}^\mathrm{(ReM\text{-}Lips)}(\mathbf{y})\|_2 \leq \|\mathbf{x}-\mathbf{y}\|_2 + \|(\mathcal{E}(\mathbf{x}))_+-(\mathcal{E}(\mathbf{y}))_+\|_2 \leq (1+\operatorname{Lip}(\mathcal{E}))\,\|\mathbf{x}-\mathbf{y}\|_2$, and hence $\operatorname{Lip}(\mathcal{D}_{\mathcal{E}}^{(\mathrm{ReM\text{-}Lips)}})\le 1+\operatorname{Lip}(\mathcal{E})$. Comparing these bounds, we have $\operatorname{Lip}(\mathcal{D}_{\mathcal{E}}^{(\mathrm{ReM\text{-}Lips)}})\leq1+\operatorname{Lip}(\mathcal{E})\leq\sqrt{2}\,\sqrt{(\operatorname{Lip}(\mathcal{E}))^2+1}$, i.e., the bound in Proposition~\ref{prop:bound} is loose.

\section{Proof of Proposition \ref{prop:lip_equal}}
\label{app:proof_prop_lip_equal}

\begin{proof}
We first prove that $\operatorname{Lip}(\mathcal{D}_{\!\mathcal{A}})\le \operatorname{Lip}(\mathcal{A})$. In Eq.~\eqref{eq:cond1}, we may set $L_1=\operatorname{Lip}(\mathcal{A})$. Moreover, as given in the last sentence of the proof of Theorem~\ref{prop:zero_preserving}, Eq.~\eqref{eq:cond2} also holds with $L_2=L_1$. The argument in the proof of Theorem~\ref{thm:lipsam} then gives $\operatorname{Lip}(\mathcal{D}_{\!\mathcal{A}})\le\max(L_1,L_2)=L_1=\operatorname{Lip}(\mathcal{A})$.

Conversely, we show that $\operatorname{Lip}(\mathcal{A})\le\operatorname{Lip}(\mathcal{D}_{\!\mathcal{A}})$.
$\mathcal{D}_{\!\mathcal{A}}$ coincides with $\mathcal{A}$ on $\mathbb{R}_+^N\subset\mathbb{C}^N$. Hence, for any distinct $\mathbf{x},\mathbf{y}\in\mathbb{R}_+^N\subset\mathbb{C}^N$, we have $\mathcal{D}_{\!\mathcal{A}}(\mathbf{x})=\mathcal{A}(\mathbf{x})$ and $\mathcal{D}_{\!\mathcal{A}}(\mathbf{y})=\mathcal{A}(\mathbf{y})$, and therefore $\|\mathcal{A}(\mathbf{x})-\mathcal{A}(\mathbf{y})\|_2\le\operatorname{Lip}(\mathcal{D}_{\!\mathcal{A}})\|\mathbf{x}-\mathbf{y}\|_2$,
meaning that $\operatorname{Lip}(\mathcal{A})\le\operatorname{Lip}(\mathcal{D}_{\!\mathcal{A}})$ holds.
Combining the two inequalities gives $\operatorname{Lip}(\mathcal{A})=\operatorname{Lip}(\mathcal{D}_{\!\mathcal{A}})$.
\end{proof}

\section{Proof of Theorem \ref{thm:worst_case_lipschitz_bound}}
\label{proof:thm_wb}
To show Theorem \ref{thm:worst_case_lipschitz_bound}, we provide the following two lemmas.

\begin{lem}
\label{lem:lips_bound_diag_sup}
Let $\mathbf{D}_1=\mathrm{diag}(a_1,\dots,a_N)\in\mathbb{R}^{N\times N}$,
$\mathbf{D}_2=\mathrm{diag}(b_1,\dots,b_N)\in\mathbb{R}^{N\times N}$ and $L\ge 0$.
Then
\begin{align*}
\sup_{\|\mathbf{J}\|_{\mathrm{op}}\le L}\ \|\mathbf{D}_1 \mathbf{J}+\mathbf{D}_2\|_{\mathrm{op}}
=
\sup_{\|\mathbf{s}\|_2\le L}\sup_{\|\mathbf{t}\|_2=1}\|\mathbf{D}_1 \mathbf{s}+\mathbf{D}_2 \mathbf{t}\|_2,
\end{align*}
where $\mathbf{J}\in\mathbb{R}^{N\times N}$ and $\mathbf{t},\mathbf{s}\in\mathbb{R}^N$.
\end{lem}

\begin{proof}
Fix any $\mathbf{J}\in\mathbb{R}^{N\times N}$ such that $\|\mathbf{J}\|_{\mathrm{op}}\le L$, and define
$\mathbf{A}=\mathbf{D}_1\mathbf{J}+\mathbf{D}_2$.
By the definition of the operator norm, there exists a unit vector $\mathbf{t}^\star\in\mathbb{R}^N$ (a right singular vector of $\mathbf{A}$) such that
$\|\mathbf{A}\|_{\mathrm{op}}=\|\mathbf{A}\mathbf{t}^\star\|_2$.
Let $\mathbf{s}^\star=\mathbf{J}\mathbf{t}^\star$. Then
$\|\mathbf{s}^\star\|_2=\|\mathbf{J}\mathbf{t}^\star\|_2\le \|\mathbf{J}\|_{\mathrm{op}}\|\mathbf{t}^\star\|_2\le L$.
Therefore, $\|\mathbf{A}\|_{\mathrm{op}}
=\|\mathbf{A}\mathbf{t}^\star\|_2
=\|\mathbf{D}_1(\mathbf{J}\mathbf{t}^\star)+\mathbf{D}_2\mathbf{t}^\star\|_2 =\|\mathbf{D}_1\mathbf{s}^\star+\mathbf{D}_2\mathbf{t}^\star\|_2
\leq \sup_{\|\mathbf{s}\|_2\leq L}\sup_{\|\mathbf{t}\|_2=1}\|\mathbf{D}_1\mathbf{s}+\mathbf{D}_2\mathbf{t}\|_2$.
Taking the supremum over all $\mathbf{J}$ with $\|\mathbf{J}\|_{\mathrm{op}}\le L$ yields
$\sup_{\|\mathbf{J}\|_{\mathrm{op}}\le L}\ \|\mathbf{D}_1\mathbf{J}+\mathbf{D}_2\|_{\mathrm{op}}
\le
\sup_{\|\mathbf{s}\|_2\leq L}\sup_{\|\mathbf{t}\|_2=1}\|\mathbf{D}_1\mathbf{s}+\mathbf{D}_2\mathbf{t}\|_2$.

Conversely, fix any $\mathbf{s},\mathbf{t}\in\mathbb{R}^N$ such that $\|\mathbf{s}\|_2\leq L$ and $\|\mathbf{t}\|_2=1$, and define
$\widetilde{\mathbf{J}}=\mathbf{s}\mathbf{t}^\mathsf{T}$.
Then, $\widetilde{\mathbf{J}}\mathbf{t}=\mathbf{s}$ and $\|\widetilde{\mathbf{J}}\|_{\mathrm{op}}=\|\mathbf{s}\mathbf{t}^\mathsf{T}\|_{\mathrm{op}}\leq L$.
Hence, $\|\mathbf{D}_1\widetilde{\mathbf{J}}+\mathbf{D}_2\|_{\mathrm{op}}
\ge \|(\mathbf{D}_1\widetilde{\mathbf{J}}+\mathbf{D}_2)\mathbf{t}\|_2
= \|\mathbf{D}_1\mathbf{s}+\mathbf{D}_2\mathbf{t}\|_2$.
Taking the supremum over $\mathbf{s}$ and $\mathbf{t}$ gives
$\sup_{\|\mathbf{J}\|_{\mathrm{op}}\le L}\|\mathbf{D}_1\mathbf{J}+\mathbf{D}_2\|_{\mathrm{op}}
\ge
\sup_{\|\mathbf{s}\|_2\le L}\sup_{\|\mathbf{t}\|_2=1}\|\mathbf{D}_1\mathbf{s}+\mathbf{D}_2\mathbf{t}\|_2$.
\end{proof}

\begin{lem}
\label{lem:lips_bound_diag}
Consider the following optimization problem:
\begin{align}
\max_{\mathbf{s},\mathbf{t}\in\mathbb{R}^N}
\|\mathbf{D}_1\mathbf{s}+\mathbf{D}_2\mathbf{t}\|_2^2
\quad\text{\rm{subject to}}\quad
\|\mathbf{s}\|_2=\|\mathbf{t}\|_2=1,
\label{eq:opt_prob}
\end{align}
where $\mathbf{D}_1=\mathrm{diag}(a_1,\dots,a_N)\in\mathbb{R}^{N\times N}$ and
$\mathbf{D}_2=\mathrm{diag}(b_1,\dots,b_N)\in\mathbb{R}^{N\times N}$.
Then there exists a solution $(\mathbf{s}^\star,\mathbf{t}^\star)$ with $(i,j)\in[N]^2$ such that
\begin{align}
(s_k^\star,t_k^\star)=(0,0)
\text{ for all }k\notin\{i,j\}.
\label{eq:twoelement}
\end{align}
\end{lem}
\begin{proof}
If either $\mathbf{D}_1$ or $\mathbf{D}_2$ is zero, the result is trivial.
Hence, we assume $\mathbf{D}_1$, $\mathbf{D}_2\neq\mathbf{O}$. 
Since Eq.~\eqref{eq:opt_prob} is a maximization problem,
we may assume that $a_n,b_n,s_n,t_n\geq 0$ for all $n\in[N]$.
Let $(\tilde{\mathbf{s}},\tilde{\mathbf{t}})$ be an optimal solution to Eq.~\eqref{eq:opt_prob} and let $\tilde{r}_n=(a_n\tilde{s}_n+b_n\tilde{t}_n)^2$ and define a set of indices as $I=\{n\in[N]\mid \tilde{r}_n\neq0\}$. 
The Lagrangian of Eq.~\eqref{eq:opt_prob} is
\begin{equation*}
\mathfrak{L}(\mathbf{s},\mathbf{t},\lambda,\mu)
=
\sum_{n=1}^N(a_ns_n+b_nt_n)^2
+\lambda(1-\|\mathbf{s}\|_2^2)
+\mu(1-\|\mathbf{t}\|_2^2).
\end{equation*}
The stationarity conditions, i.e., $\nabla_\mathbf{s}\mathfrak{L}(\tilde{\mathbf{s}},\tilde{\mathbf{t}},\lambda,\mu)=\mathbf{0}$ and $
\nabla_\mathbf{t}\mathfrak{L}(\tilde{\mathbf{s}},\tilde{\mathbf{t}},\lambda,\mu)=\mathbf{0}$, give $a_n \sqrt{\tilde{r}_n}=\lambda \tilde{s}_n$, $b_n \sqrt{\tilde{r}_n}=\mu \tilde{t}_n$ for all $n\in[N]$. 
Since $\mathbf{D}_1,\mathbf{D}_2\neq\mathbf{O}$ and
$(\tilde{\mathbf{s}},\tilde{\mathbf{t}})$ is optimal, we have
$\mathbf{D}_1\tilde{\mathbf{s}}\neq\mathbf{0}$ and
$\mathbf{D}_2\tilde{\mathbf{t}}\neq\mathbf{0}$, and hence
$\lambda,\mu>0$.
Therefore, $\tilde{r}_n=0$ implies $\tilde{s}_n=\tilde{t}_n=0$.
Let $c_n=a_n^2/\lambda^2$ and $d_n=b_n^2/\mu^2$.
Substituting $\tilde{s}_n=\sqrt{c_n\tilde{r}_n}$ and $\tilde{t}_n=\sqrt{d_n\tilde{r}_n}$ into
$\sqrt{\tilde{r}_n}=a_n\tilde{s}_n+b_n\tilde{t}_n$ gives $\lambda c_n+\mu d_n=1$ for all $n\in I$. 
Let us consider the following linear programming (LP):
\begin{align}
\max_{\mathbf{r}\in\mathbb{R}_+^{|I|}}
\sum_{n\in I}r_n
\;\text{subject to}\;
\mathbf{P}
\mathbf{r}
=\mathbf{1},\;
\mathbf{P}=
\begin{pmatrix}
(c_n)_{n\in I}\\
(d_n)_{n\in I}
\end{pmatrix},
\label{eq:theLP}
\end{align}
where $\mathbf{P}\in\mathbb{R}_+^{2\times|I|}$.
Then $(\tilde{r}_n)_{n\in I}\in\mathbb{R}_+^{|I|}$ is feasible to Eq.~\eqref{eq:theLP}, since the constraints in Eq.~\eqref{eq:opt_prob} give $\sum_{n\in I}\tilde{s}_n^2=\sum_{n\in I}c_n\tilde{r}_n=1$ and $\sum_{n\in I}\tilde{t}_n^2=\sum_{n\in I}d_n\tilde{r}_n=1$. Its objective value also coincides with that of Eq.~\eqref{eq:opt_prob}, since $\sum_{n\in I}\tilde{r}_n=\|\mathbf{D}_1\tilde{\mathbf{s}}+\mathbf{D}_2\tilde{\mathbf{t}}\|_2^2$.
Moreover, every feasible point for Eq.~\eqref{eq:theLP} is optimal, i.e., all feasible points have the same objective value, $\lambda+\mu$. 
Indeed, for any $\mathbf{r}\in\mathbb{R}_+^{|I|}$ such that $\mathbf{P}\mathbf{r}=\mathbf{1}$, we have $\sum_{n\in I}r_n=\sum_{n\in I}(\lambda c_n+\mu d_n)r_n=\lambda\sum_{n\in I}c_nr_n+\mu\sum_{n\in I}d_nr_n=\lambda+\mu$. Since the LP admits a basic feasible solution and $\operatorname{rank}(\mathbf{P})\leq2$, there exists an optimal solution $\mathbf{r}^\star$ with at most two nonzero elements. A solution $(\mathbf{s}^\star,\mathbf{t}^\star)$ to Eq.~\eqref{eq:opt_prob} can then be constructed from $\mathbf{r}^\star$ by setting $(s_n^\star,t_n^\star)=(\sqrt{c_n r_n^\star},\sqrt{d_nr_n^\star})$ for $n\in I$, and $(s_n^\star,t_n^\star)=(0,0)$ otherwise, which satisfies Eq.~\eqref{eq:twoelement}.
\end{proof}

Using these lemmas, we show Theorem \ref{thm:worst_case_lipschitz_bound}.

\begin{proof}[Proof of Theorem \ref{thm:worst_case_lipschitz_bound}]
Let $\Lambda_\psi = \{(a,b)\in\mathbb{R}^2 \,|\,
\exists (u,v)\in\mathbb{R}\times\mathbb{R}_+\text{ s.t. } a=\partial_1\psi(u,v), b=\partial_2\psi(u,v)
\}$ be the set of all possible pairs of the partial derivatives.
For any $N\geq 1$ and $\mathbf{x}\in\mathbb{R}_+^N$,
there exist $(a_n,b_n)\in\Lambda_\psi$ for all $n\in[N]$ such that
$\mathbf{J}_{\mathcal{A}_{(\mathcal{L},\psi)}}(\mathbf{x})=\mathbf{D}_1\,\mathbf{J}_\mathcal{L}(\mathbf{x})+\mathbf{D}_2$, where $\mathcal{A}_{(\mathcal{L},\psi)}(\mathbf{x}) = (\psi((\mathcal{L}(\mathbf{x}))_n, x_n))_{n=1}^N$,
$\mathbf{D}_1=\mathrm{diag}(a_1,\dots,a_N)$,
and $\mathbf{D}_2=\mathrm{diag}(b_1,\dots,b_N)$.
When $N\geq 2$, Lemmas~\ref{lem:lips_bound_diag_sup} and \ref{lem:lips_bound_diag} imply
\begin{align*}
&B_{(L,\psi)}^{(N)}=\sup_{\mathcal{L}:\mathbb{R}_+^N\to\mathbb{R}^N:\operatorname{Lip}(\mathcal{L})\le L}\operatorname{Lip}(\mathcal{A}_{(\mathcal{L},\psi)})\nonumber \\
&=
\sup_{((a_n,b_n))_{n=1}^N\in\Lambda_\psi^N}\ 
\sup_{\|\mathbf{J}\|_{\mathrm{op}}\le L}\ \|\mathbf{D}_1\mathbf{J}+\mathbf{D}_2\|_{\mathrm{op}} \nonumber \\
&=\sup_{((a_n,b_n))_{n=1}^N\in\Lambda_\psi^N}    \sup_{\scriptsize\begin{array}{c}
                \mathbf{t}\in\mathbb{R}^N\\
                \|\mathbf{t}\|_2=1
                \end{array}
    }
        \sup_{\scriptsize\begin{array}{c}
                \mathbf{s}\in\mathbb{R}^N\\
                \|\mathbf{s}\|_2\leq 1
                \end{array}
    } \|(L\mathbf{D}_1)\mathbf{s}+\mathbf{D}_2\mathbf{t}\|_2 \nonumber \\
&\overset{\star}{=} \sup_{((\hat{a}_n,\hat{b}_n))_{n=1}^2\in\Lambda_\psi^2}\ 
    \sup_{\scriptsize\begin{array}{c}
                \hat{\mathbf{t}}\in\mathbb{R}^2\\
                \|\hat{\mathbf{t}}\|_2=1
                \end{array}
    }
    \sup_{\scriptsize\begin{array}{c}
                \hat{\mathbf{s}}\in\mathbb{R}^2\\
                \|\hat{\mathbf{s}}\|_2\leq 1
                \end{array}
    }
    \|(L\hat{\mathbf{D}}_1)\hat{\mathbf{s}}+\hat{\mathbf{D}}_2\hat{\mathbf{t}}\|_2, 
    \\
    &= \sup_{\mathcal{L}:\mathbb{R}_+^2\to\mathbb{R}^2:\operatorname{Lip}(\mathcal{L})\le L}\operatorname{Lip}(\mathcal{A}_{(\mathcal{L},\psi)})=B_{(L,\psi)}^{(2)}.
\end{align*}
The equality marked by $\star$ follows by Lemma~\ref{lem:lips_bound_diag}.
Namely, the supremum in the upper line is attained by some $\mathbf{s}^\star,\mathbf{t}^\star\in\mathbb{R}^N$ for which there exists $(i,j)\in[N]^2$ such that $(s^\star_k,t^\star_k)=(0,0)$ for all $k\notin\{i,j\}$. 
Therefore, retaining only the $i$th and $j$th components and the corresponding diagonal entries of $\mathbf{D}_1$ and $\mathbf{D}_2$, denoted by hats, preserves the objective value while reducing the dimensions of $\mathbf{s}$ and $\mathbf{t}$ to $N=2$.
For $N=1$, any $L$-Lipschitz-continuous mapping $\mathcal{L}_1:\mathbb{R}_+\to\mathbb{R}$ can be extended to $\mathcal{L}_2:\mathbb{R}_+^2\to\mathbb{R}^2$ by $\mathcal{L}_2(x,y)=(\mathcal{L}_1(x),\mathcal{L}_1(y))$, for which $\operatorname{Lip}(\mathcal{A}_{(\mathcal{L}_2,\psi)})=\operatorname{Lip}(\mathcal{A}_{(\mathcal{L}_1,\psi)})$. Hence, $B_{(L,\psi)}^{(1)}\leq B_{(L,\psi)}^{(2)}$.
Thus, Eq.~\eqref{eq:reduce2} holds.
\end{proof}

\section{Proof of Proposition \ref{thm:lips_bound_se_re}}
\label{app:proof_thm_lips_bound_se_re}

\begin{proof}
$\operatorname{Lip}(\mathcal{D}_{\mathcal{E}}^\mathrm{(Lips)}) \leq \sqrt{(\operatorname{Lip}(\mathcal{E}))^2+1}$ follows immediately from Proposition~\ref{prop:bound}, whereas $\operatorname{Lip}(\mathcal{D}_{\mathcal{E}}^\mathrm{(ReM\text{-}Lips)}) \leq \operatorname{Lip}(\mathcal{E})+1$ follows from the triangle inequality, as shown in Appendix~\ref{sec:loose}.

We show the tightness of these bounds by providing examples that achieve equality for $N=2$.
Recall that the derivative of ReLU (i.e., $(\cdot)_+$) is $1$ for positive inputs.
First, assume that $L>0$.
For LipsAM-E, let $\mathcal{E}(x_1,x_2)=(Lx_2,Lx_1)$ and $\mathcal{A}(\mathbf{x})=(\min(\mathcal{E}(\mathbf{x}),\mathbf{x}))_+$. Then $\operatorname{Lip}(\mathcal{E})=L$, and for any $\mathbf{x}^\star=(1,x_2^\star)\in\mathbb{R}_+^2$ with $0<x_2^\star<\min(1/L,L)$, we obtain
\begin{equation}
\left\|\mathbf{J}_{\mathcal{A}}(\mathbf{x}^\star)\right\|_{\mathrm{op}}
=
\left\|
\begin{bmatrix}
0 & L\\
0 & 1
\end{bmatrix}
\right\|_{\mathrm{op}}
=
\sqrt{L^2+1}.
\end{equation}
For ReM-LipsAM-E, let $\mathcal{E}(x_1,x_2)=(b-Lx_1,b-Lx_2)$ for some $b>0$
and $\mathcal{A}(\mathbf{x})=(\mathbf{x}-(\mathcal{E}(\mathbf{x}))_+)_+$.
Then $\operatorname{Lip}(\mathcal{E})=L$, and for any $\mathbf{x}^\star=(x^\star,x^\star)\in\mathbb{R}_+^2$ satisfying $b/(L+1)<x^\star<b/L$,
we obtain 
\begin{equation}
\left\|
\mathbf{J}_{\mathcal{A}}
(\mathbf{x}^\star)
\right\|_{\mathrm{op}}
=
\left\|
\begin{bmatrix}
L+1 & 0\\
0 & L+1
\end{bmatrix}
\right\|_{\mathrm{op}}
=
L+1.
\end{equation}
The case $L=0$ follows similarly by using a positive constant mapping, e.g., $\mathcal{E}(x_1,x_2)=(b,b)$ for some $b>0$.
\end{proof}

\bibliographystyle{IEEEtran}
\bibliography{LipsAM.bib}
\end{document}